\documentclass[sigconf,authordraft=false,review=false]{acmart}
\usepackage{sigirstyle}
\usepackage{appendix}
\usepackage{balance}
\usepackage{amsmath}
\usepackage{booktabs}
\usepackage{tabularx}
\usepackage{caption}
\usepackage{subcaption}
\usepackage{adjustbox}
\usepackage{caption}
\usepackage{hyperref}

\usepackage[utf8]{inputenc}
\usepackage{url}

\usepackage{calrsfs}

\DeclareMathAlphabet{\pazocal}{OMS}{zplm}{m}{n}
\DeclareMathAlphabet{\pazobfcal}{OMS}{cmsy}{b}{n}

\copyrightyear{2026}
\acmYear{2026}
\setcopyright{cc}
\setcctype{by}
\acmConference[ICTIR '26] {Proceedings of the 2026 International ACM SIGIR Conference on Innovative Concepts and Theories in Information Retrieval (ICTIR)}{July 25, 2026}{Melbourne, VIC, Australia.}
\acmBooktitle{Proceedings of the 2026 International ACM SIGIR Conference on Innovative Concepts and Theories in Information Retrieval (ICTIR) (ICTIR '26), July 25, 2026, Melbourne, VIC, Australia}
\acmISBN{979-8-4007-2600-2/2026/07}
\acmDOI{10.1145/3805713.3820421}

\begin{document}


\title{A Theoretical Framework for Risk Analysis of Stochastic Rankers}

\author{Debasis Ganguly}
\affiliation{%
  \institution{University of Glasgow}
  \city{Glasgow}
  \country{United Kingdom}
}
\email{debasis.ganguly@glasgow.ac.uk}

\begin{abstract}

Different from deterministic rankers that seek to maximize relevance at top ranks,
stochastic ranking policies instead estimate distributions over permutations, from
which rankings are sampled, towards obtaining diversified or fair exposure.
Such policies are commonly evaluated in terms of expected effectiveness post-reranking.
However, the randomness inherent in these policies gives rise to a fundamental but
under-explored \emph{ex ante} question: prior to applying stochastic reranking, how large
can the induced variation in retrieval effectiveness be in the worst case?
This paper presents a theoretical analysis of \emph{reranking risk}, defined as the maximum
absolute change in discounted cumulative gain (DCG) resulting from a permutation sampled
from a stochastic reranking policy applied to a fixed retrieved list.
We derive that this risk is governed by the distribution of the recall points in the
initial retrieved list.
We conduct experiments on submitted runs from the TREC Fairness 2022 track that
employ stochastic reranking policies and empirically demonstrate that the effectiveness
variations predicted by our theory closely approximate the observed changes in DCG.

%
%

\end{abstract}


\maketitle


\section{Introduction}
\label{sec:intro}

Unlike classic ranking algorithms that produce a single, static
ranking for a given query, a stochastic ranking policy provides a distribution over all possible rankings of
candidate items \cite{diazcikm2020,Yadav2019PolicyGradientTO}.
%
%
Stochastic ranking policies are motivated by the objective of decoupling relevance-based
retrieval from the presentation of a single, fixed ranking to all users \cite{Heuss2022FairnessOE,diazcikm2020}.
Rather than committing to one deterministic ordering of retrieved documents for a given
query, such policies define a distribution over rankings, from which different orderings
may be presented across users or interactions \cite{Bower2021IndividuallyFR}.
This enables exposure to be allocated across documents in a controlled yet
non-deterministic manner.

%
For instance, a deterministic ranking produced for a query such as ``successful CEOs'' may lead to a disproportionate representation of demographic groups even in
the presence of equally relevant alternatives \cite{DBLP:conf/aaai/SenG20}.
By sampling from a distribution over plausible rankings, stochastic policies allow
for a more balanced allocation of exposure while maintaining relevance in expectation \cite{Biega2018EquityOA,DBLP:conf/trec/EkstrandMR022}.
For this use-case, an ideal stochastic ranker is likely to occasionally promote a relevant document describing a female CEO from an under-represented demographic group to a
higher rank, while demoting another relevant document describing a male CEO from a
historically over-represented group.
However, since document relevance is not directly observed but instead estimated as a
posterior conditioned on the query and available signals, such rank transpositions are not
without risk.
In particular, stochastic promotion may inadvertently elevate non-relevant documents,
thereby degrading overall retrieval effectiveness.
This introduces an intrinsic tension between exploiting stochasticity to improve
fairness-related objectives and preserving effectiveness, motivating the need for
principled risk-aware analysis.

In deterministic rerankers, this tension manifests itself in a different form. To be more precise, as
given a query and its retrieved documents, a deterministic ranker induces a single,
realized ordering~\cite{mono-duo-lin}; it is thus possible to ask whether that ranking is
likely to perform well.
In fact, this question has been the subject of extensive study under the umbrella of query
performance prediction (QPP), which aims to estimate retrieval effectiveness without
relevance judgments by exploiting signals derived from the query, the retrieved result
set, and the associated retrieval scores~\cite{kurland_tois12,uef_kurland_sigir10,qpp-bert,datta2022deep}.

Incorporating non-determinism in ranking makes it more challenging to answer the same very question that QPP answers for deterministic rankers, namely, \textit{how likely is the ranking to perform well}. This is because unlike deterministic rankers where the rank of a document is induced by its similarity score with the query, for stochastic rankers
the rank of any document in the reranked list is itself a random variable having little or no dependence on the similarity score.
%

The key question we study in this paper is not how well a stochastic reranker performs on average after repeated sampling \cite{diazcikm2020}, but rather the ex ante analysis of how may the retrieval effectiveness of a given retrieved list change due to the rank transpositions applied on it.
We formalize this quantity as \emph{reranking risk}
and we derive that this risk is determined by
two structural factors: the locations of the relevant documents in the initial ranking,
and the extent to which the policy allows documents to move across rank positions.
This theoretical finding of ours that the risk of reranking depends on the recall points and the aggressiveness of the rank shifts is intuitive because changes in IR effectiveness
can arise only from the rank transpositions involving relevant documents.

From a practical perspective, the ability to estimate reranking risk is particularly
relevant for stochastic rankers that are optimized with respect to objectives
complementary to relevance, such as fairness or diversity.
In such settings, not all rank transpositions contribute equally to effectiveness
variation: permutations that induce large displacements of relevant documents are
inherently more risky than those that preserve their rank locality.
A risk-aware analysis can potentially be used to discard the permutations with high predicted risks, thus improving overall performance measures such as expected exposure relevance \cite{diazcikm2020}.

A fundamental difference between the reranking risk estimation and QPP arises from the fact that while the former depends on the recall points, the latter makes no use of relevance assessments. While this may appear to be a limitation, a practical solution is to leverage document clicks as proxies for relevance or apply click models trained on past data to predict document clicks~\cite{joachims2005accurately,craswell2008experimental,chapelle2009dynamic}. 

The rest of the paper is organized as follows. Section \ref{sec:related} provides a brief background on stochastic ranking and query performance prediction. In Section \ref{sec:theory}, we formally derive the expected risk of reranking for a relatively simple case, i.e., when only a single relevant document is retrieved in the initial list. We later generalise the theory to address multiple relevant documents in Section \ref{sec:dcg_multi}. The theoretical analysis in both Sections \ref{sec:theory} and \ref{sec:dcg_multi} model a stochastic policy first as a uniform distribution that considers any rank swap equally likely, and second as a kernel-based distribution that assigns higher likelihood to local rank swaps. Section \ref{sec:experimental-setup} presents an empirical validation of the theoretical analysis for official runs submitted to the TREC fairness 2022 task. Finally, Section \ref{sec:conclusions} concludes the paper with directions for future work.

\section{Related Work}
\label{sec:related}

\para{Stochastic Ranking}

Stochastic ranking policies consist of two main components: a scoring model (typically a neural network such as \cite{mono-duo-lin}) to generate the
relevance scores and the initial list, and a sampling method that converts these scores into probability distributions over rankings.
The sampling method defines how to create a
probability distribution over all possible rankings based on these scores, some of which 
are listed below.
\uls
\li \textbf{Plackett-Luce (PL) based Ranking Models}: This is the most widely adopted approach for stochastic ranking policies, using Luce's axiom of choice to iteratively sample items without replacement with probabilities defined by softmax distributions over relevance scores \cite{Wu2022JointME,Togashi2022FairMF,Oosterhuis2021ComputationallyEO,Singh2019PolicyLF,Guo2023InferencetimeSR,Togashi2022FairMF}.

\li \textbf{Birkhoff-von Neumann Decomposition}: Converts linear programming formulations into practical stochastic ranking policies by decomposing doubly stochastic matrices into convex combinations of permutation matrices \cite{Vardasbi2022ProbabilisticPG}.

\li \textbf{Stochastic Gradient Descent with Differentiable Sampling}: Methods that replace deterministic score-to-permutation mappings with differentiable sampling from distributions defined by raw scores, improving model robustness during training \cite{Gorantla2023OptimizingGP,Bruch2020AST}.
\ule

\para{Query Performance Prediction (QPP)}

QPP aims to estimate retrieval effectiveness in the
absence of relevance judgments, often by modelling uncertainty in relevance distributions.
Traditional unsupervised QPP approaches rely on score-based statistics~\cite{kurland_tois12,wig_croft_SIGIR07}
and aggregation over query variants or reference lists~\cite{query_variants_kurland,DBLP:journals/tois/DattaGMG23,uef_kurland_sigir10,rsd_haggai}.
In contrast to score-based methods, embedding-based approaches leverage information from the local neighbourhoods of query and document representations to estimate retrieval quality~\cite{RoyEtAl2019,FaggioliFerroEtAl2023,pdd-tois}. 
Supervised QPP approaches, on the other hand, employ end-to-end neural models trained on query-document content features to predict retrieval effectiveness~\cite{qpp-bert,datta2022pointwise,datta2022deep}.

Beyond its conventional use for improving downstream retrieval performance~\cite{DBLP:conf/ecir/DattaGMG24}, QPP has recently evolved in two prominent directions.
First, in multi-stage retrieval pipelines, QPP can be used to selectively invoke computationally intensive downstream components when the predicted effectiveness of earlier stages is deemed insufficient~\cite{DBLP:conf/ecir/SantraBG26,DBLP:conf/ecir/SantraBG26a}.
Second, in the context of retrieval-augmented generation (RAG) for tasks such as complex question answering and exploratory search, recent work has examined the relationship between document relevance and its utility within a generative pipeline~\cite{DBLP:conf/ecir/TianGM25}.
This has led to the development of methods for predicting document utility, enabling more adaptive and efficient RAG workflows~\cite{rpp-gpp,santra2025hf}.



\section{Reranking Risk with one Relevant Document}
\label{sec:theory}

We begin our theoretical analysis with a simple setting in which exactly one relevant document appears in the initially retrieved list.
Despite its simplicity, this case is useful to reveal how the reranking risk depends on two factors: i) rank of the relevant document in the retrieved list, and ii) the characteristic of the ranking policy itself, i.e., how aggressively the stochastic ranker shifts the initial rank positions.
Without assuming any a-priori knowledge on the ranking policy used, we approximate the characteristics of a stochastic ranker by two permutation probability distributions -- i) uniform,  that assigns an equal likelihood to any rank transposition, and ii) a locality-biased distribution that assigns a higher likelihood to local rank swaps similar to Plackett-Luce \cite{Oosterhuis2021ComputationallyEO}.

\subsection{Problem Formulation}
\label{ss:setup-rr}
Let $Q$ denote a query and let $L_M(Q) = (d_1,\ldots,d_M)$
be the ranked list of the top-$M$ documents retrieved for $Q$, where the index indicates rank position.
We study the effect of reranking on the IR effectiveness of $L_M(Q)$.
A reranking operation is modeled as a permutation $\sigma \in S_M$ acting on the rank positions, yielding the perturbed ranked list $\sigma(L_M(Q))$, which we denote as $L^\sigma_M(Q)$ for brevity.

Let $\mu(Q; L_M(Q))$
denote the effectiveness of $L_M(Q)$ under a target IR metric $\mu$.
Because most IR metrics are rank-based (rather than set-based), applying a permutation $\sigma$
generally induces a change in effectiveness, i.e.,
$\mu(Q; L_M(Q)) \neq \mu(Q; L^\sigma_M(Q))$.    
%
The reranking-induced effectiveness change can now be written as:
\begin{equation}
\Delta_\mu(Q,\sigma)
=
\left|
\mu(Q; L_M(Q)) - \mu(Q; L^\sigma_M(Q))
\right|,
\label{eq:delta_mu_def}
\end{equation}
which quantifies the risk introduced by reranking a given ranked list for query $Q$.

\subsection{Rank Swaps with Uniform Likelihood}

For the analysis in this section, we now assume that only a single document in the ranked list $L_M(Q)$ is relevant. We denote the rank of this relevant document as a random variable $K \in \{1,\ldots,M\}$, which means that $\exists K=k: d_k \in \pazocal{R}(Q)$, where $\pazocal{R}(Q)$ denotes the set of documents that are relevant to the query $Q$.
The suitable metric to consider in Equation \ref{eq:delta_mu_def} under the assumption of only one relevant document is the reciprocal rank (RR), i.e.,
\begin{equation}
\mu(Q; L_M(Q)) = \mathrm{RR}(Q; L_M(Q)) = \frac{1}{k} \notag.
\end{equation}
Let $\sigma \in S_M$ be a permutation applied to the ranked list.
The rank where the relevant document moves to is also a random variable of the form $R=\sigma(K)$. For a particular observation $R=r$ of this random variable, the value of the RR metric changes to: 
$$
\mathrm{RR}(Q; L^\sigma_M(Q)) = 1/r.
$$
We define the permutation-induced RR metric change as
\begin{equation}
\Delta_{\mathrm{RR}}(Q,\sigma)
=
\left|
\mu(Q; L_M(Q)) - \mu(Q; L^\sigma_M(Q))
\right|
=
\left|\frac{1}{k} - \frac{1}{r}\right|.
\label{eq:rr_gap}
\end{equation}

For uniform rank perturbations, the new rank $R$ of the relevant document is drawn uniformly
from $\{1,\ldots,M\} \setminus \{k\}$:
\[
\Pr(R=r \mid K=k) = \frac{1}{M-1}, \quad r \neq k.
\]
The expected change in reciprocal rank induced by a uniform rank perturbation can be
written as
\begin{equation}
\mathbb{E}\!\left[\Delta_{\mathrm{RR}} \mid K = k\right]
=
\frac{1}{M-1}
\sum_{r \neq k}
\left|
\frac{1}{k} - \frac{1}{r}
\right|.
\label{eq:expectederror-start}
\end{equation}

To evaluate the sum, we split it into two parts corresponding to ranks above and
below $k$.
For $r < k$, we have $\left| \frac{1}{k} - \frac{1}{r} \right| = \frac{1}{r} - \frac{1}{k}$,
whereas for $r > k$, we have $\left| \frac{1}{k} - \frac{1}{r} \right| = \frac{1}{k} - \frac{1}{r}$.
Substituting these expressions into Equation~\ref{eq:expectederror-start} yields
\begin{equation}
\mathbb{E}\!\left[\Delta_{\mathrm{RR}} \mid K = k\right]
=
\frac{1}{M-1}
\left(
\sum_{r=1}^{k-1}
\left( \frac{1}{r} - \frac{1}{k} \right)
+
\sum_{r=k+1}^{M}
\left( \frac{1}{k} - \frac{1}{r} \right)
\right).
\label{eq:expectederror-split}
\end{equation}

We now simplify each term separately.
The first sum can be written as
\[
\sum_{r=1}^{k-1} \left( \frac{1}{r} - \frac{1}{k} \right)
=
\sum_{r=1}^{k-1} \frac{1}{r}
-
\frac{k-1}{k}
=
H_{k-1} - \frac{k-1}{k},
\]
where $H_n = \sum_{i=1}^n \frac{1}{i}$ denotes the $n$-th harmonic number.
Similarly, the second sum satisfies
\[
\sum_{r=k+1}^{M} \left( \frac{1}{k} - \frac{1}{r} \right)
=
\frac{M-k}{k}
-
\sum_{r=k+1}^{M} \frac{1}{r}
=
\frac{M-k}{k}
-
\left( H_M - H_k \right).
\]
Substituting these expressions back into Equation~\ref{eq:expectederror-split} gives
\begin{equation}
\mathbb{E}\!\left[\Delta_{\mathrm{RR}} \mid K = k\right]
=
\frac{
H_{k-1}
- \frac{k-1}{k}
+ \frac{M-k}{k}
- (H_M - H_k)
}{M-1}.
\label{eq:expectederror-expanded}
\end{equation}
Using the identity $H_k = H_{k-1} + \frac{1}{k}$, we can simplify the numerator of
Equation~\ref{eq:expectederror-expanded} to obtain
\begin{equation}
\mathbb{E}\!\left[\Delta_{\mathrm{RR}} \mid K = k\right]
=
\frac{
2H_{k-1}
- H_M
+ \frac{M - 2k + 1}{k}
}{M-1},
\label{eq:expectederror-singledoc}
\end{equation}
which yields the closed-form expression used in the subsequent asymptotic analysis.
We now characterize the dominant asymptotic behaviour by the following proposition.
%



\begin{proposition}[Reranking Risk under Uniform Rank Perturbations] \label{prop:sd-uniform} Consider a ranked list containing a single relevant document at rank $K=k$. Under uniform rank swap perturbations, \[ \mathbb{E}\!\left[\Delta_{\mathrm{RR}} \mid K=k\right] = \frac{ 2H_{k-1} -H_M +\frac{M-2k+1}{k} }{M-1}. \] Furthermore, for ranks satisfying $k=\Theta(M)$, \[ \mathbb{E}\!\left[\Delta_{\mathrm{RR}} \mid K=k\right] = \Theta\!\left(\frac{\log k}{M}\right). \] \end{proposition}

\begin{figure}[t]
    \centering
    \includegraphics[width=0.9\columnwidth]{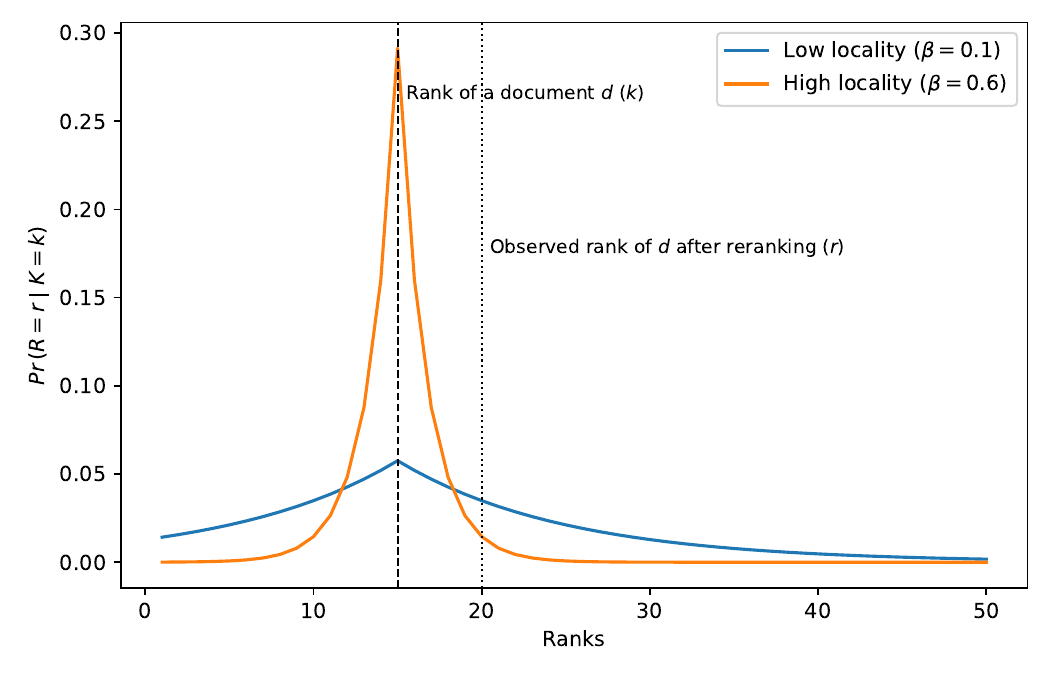}
    \caption{
    Effect of the locality parameter $\beta$ on the rank perturbation model (Equation \ref{eq:localityranker}).
    The figure shows the conditional distribution $P(R=r \mid K=k)$ for a document
    originally at rank $k$, under two different values of $\beta$.
    Smaller values of $\beta$ yield broader distributions with higher probability
    of large rank displacements, while larger values of $\beta$ concentrate
    probability mass near the original rank, corresponding to more local reranking behaviour.
    }
    \label{fig:locality-kernel}
\end{figure}

Proposition~\ref{prop:sd-uniform} shows that the sensitivity of reciprocal rank to uniform perturbations depends strongly on the initial position of the relevant document. Documents appearing near the head of the ranking experience the largest expected change, since a random transposition is likely to move them far from their highly valuable original position. For documents occurring deeper in the ranking ($k=\Theta(M)$), the expected perturbation is substantially smaller and scales as $\Theta(\log k/M)$.


\subsection{Locality-Biased Rank Perturbations} \label{ss:localrankswapmodel-sd}

Uniform rank swaps represent a worst-case perturbation model in which a relevant
document may be displaced arbitrarily far from its original position.
In practice, however, reranking mechanisms tend to induce more localized changes \cite{joachims-kdd2018,diazcikm2020}.
To capture this behaviour, we consider locality-biased rank perturbations motivated
by stochastic ranking models such as Plackett--Luce.

The Plackett--Luce (PL) model defines a probability distribution over full
permutations by associating each document with a latent utility and sampling the
ranking sequentially. While PL provides a principled probabilistic model of ranking,
its induced marginal distribution over the rank of a single document does not admit
a simple closed form. Nevertheless, a key qualitative property of PL-type models is
that documents tend to concentrate around their expected rank, with probability
decaying as the displacement from that rank increases. As a result, most sampled
permutations differ from a reference ranking through relatively small, local
reorderings rather than global reshuffles.

To make the notion of locality explicit while retaining analytical tractability, we
model the marginal effect of reranking on the rank of a single relevant document
using a symmetric distance-based kernel, specifically: a discrete Laplace (two-sided exponential) distribution (see Figure \ref{fig:locality-kernel}).
Formally speaking,
\begin{equation}
\Pr(R = r \mid K = k)
=
\frac{\exp(-\beta |r-k|)}{Z_k},
\quad r \neq k,
\label{eq:localityranker}
\end{equation}
where the normalization constant
\[
Z_k = \sum_{r \neq k} \exp(-\beta |r-k|)
\]
ensures that the distribution sums to one over all valid target ranks.
We now characterize how imposing locality on rank perturbations changes the asymptotic bound derived in Proposition \ref{prop:sd-uniform}.

\begin{proposition}[Asymptotic reranking risk under locality-biased perturbations]
\label{prop:rr_local_asymptotic}
For a ranked list $L_M(Q)$ with a single relevant document at rank $K=k$, the expected
change in reciprocal rank (RR) under the influence of a reranker favouring local swaps satisfies
\[
\mathbb{E}\left[\Delta_{\mathrm{RR}} \mid K=k\right]
=
\pazocal{O}\!\left(\frac{1}{\beta k}\right),
\]
which is independent of $M$.
\label{prop:sd-local}
\end{proposition}

The dependence on $1/k$ reflects the fact that reciprocal rank becomes less sensitive
as the relevant document moves deeper in the ranking, while the dependence on
$1/\beta$ captures the expected magnitude of locality-biased rank displacement.
Unlike the uniform perturbation model, the resulting bound does not depend on the
candidate set size $M$, since the probability mass of the locality-biased kernel is
concentrated around the original rank.

\begin{proof}
Conditioned on $K = k$, the reciprocal rank before reranking is $1/k$, and after
reranking it becomes $1/R$, where $R$ denotes the random rank of the relevant
document after reranking.
For any realization $r \neq k$, the absolute change in reciprocal rank satisfies
\[
\left| \frac{1}{k} - \frac{1}{r} \right|
=
\frac{|r-k|}{kr}.
\]
Since $r \geq 1$ for all valid rank positions, we have $kr \geq k$, and hence
\[
\frac{1}{kr} \leq \frac{1}{k}.
\]
Therefore,
\[
\left| \frac{1}{k} - \frac{1}{r} \right|
\leq
\frac{|r-k|}{k}.
\]

Taking expectations with respect to the conditional distribution of $R$ given
$K = k$, we obtain
\begin{equation}
\mathbb{E}\!\left[\Delta_{\mathrm{RR}} \mid K = k\right]
=
\mathbb{E}
\left[
\left| \frac{1}{k} - \frac{1}{R} \right|
\;\middle|\; K = k
\right]
\le
\frac{1}{k}
\mathbb{E}\!\left[ |R-k| \mid K = k \right].
\label{eq:rr-bound-step-local-corrected}
\end{equation}

We now compute $\mathbb{E}[|R-k| \mid K = k]$ under the locality-biased
rank perturbation model in Equation~\ref{eq:localityranker}, which defines
\begin{equation}
\Pr(R = r \mid K = k)
=
\frac{\exp(-\beta |r-k|)}{Z_k}.
\label{eq:rankswapmodel}
\end{equation}
Let $D = |R-k|$ denote the absolute rank displacement.
Ignoring boundary effects, which do not affect the asymptotic scaling, the distribution
of $D$ is symmetric and satisfies
\[
\Pr(D = d)
=
\frac{2 e^{-\beta d}}{Z_k},
\quad d \in \mathbb{Z}^+.
\]
The normalization constant can be written as
\[
Z_k
=
2 \sum_{d=1}^{\infty} e^{-\beta d}
=
2 \cdot \frac{e^{-\beta}}{1 - e^{-\beta}}
=
\Theta\!\left( \frac{1}{\beta} \right),
\]
where the final equivalence follows from the Taylor expansion
$1 - e^{-\beta} = \Theta(\beta)$ for $\beta > 0$
(see Appendix~\ref{app:normalization} for details).

The expected displacement is therefore
\[
\mathbb{E}[D]
=
\sum_{d=1}^{\infty} d \Pr(D=d)
=
\frac{2}{Z_k}
\sum_{d=1}^{\infty} d e^{-\beta d}.
\]
Using the standard identity for the derivative of a geometric series, we obtain
\begin{equation}
\sum_{d=1}^{\infty} d e^{-\beta d}
=
\frac{e^{-\beta}}{(1 - e^{-\beta})^2}
=
\Theta\!\left( \frac{1}{\beta^2} \right),
\end{equation}
as shown in Appendix~\ref{app:geom_sum}.
Combining this with $Z_k = \Theta(1/\beta)$ yields
\begin{equation}
\mathbb{E}[|R-k| \mid K = k]
=
\Theta\!\left( \frac{1}{\beta} \right).
\label{eq:expected-rd-localswap}
\end{equation}

Substituting Equation~\ref{eq:expected-rd-localswap} into
Equation~\ref{eq:rr-bound-step-local-corrected} gives
\[
\mathbb{E}\!\left[\Delta_{\mathrm{RR}} \mid K = k\right]
=
\pazocal{O}\!\left( \frac{1}{\beta k} \right),
\]
which completes the proof.
\end{proof}

\section{Risk with Multiple Relevant Documents}
\label{sec:dcg_multi}

We now extend the reranking risk analysis to settings with multiple relevant
documents and consider discounted cumulative gain (DCG) \cite{jarvelin2002cumulated} as the target IR
effectiveness metric.
DCG provides a natural generalization of the single-relevant-document analysis,
as it decomposes additively over individual relevant documents and assigns
rank-dependent importance through a logarithmic discount.

\subsection{Changes in DCG} \label{ss:md-notations}

In Section \ref{ss:setup-rr}, we let the random variable $K$ denote the rank of the only relevant document. In this analysis, $K$ now denotes a set of ranks where relevant documents are retrieved, i.e., $K = \{k_1,\ldots,k_m\}$ ($k_1 < k_2 \ldots < k_m$) such that $\forall k \in K,\,d_k \in \pazocal{R}(Q)$. 
Let
\[
f(x_i) = \frac{1}{\log_2(x_i+1)}
\]
be the discounted cumulative gain (DCG) contribution from the document at the $i$-th rank assuming binary relevance (a simplification step for carrying out the computations conveniently). The total DCG of the list of retrieved documents $L_M(Q)$ is then 
\begin{equation}
\mu(Q; L_M(Q)) = \mathrm{DCG}(Q; L_M(Q))
=
\sum_{i=1}^m f(k_i) \label{eq:dcg}.
\end{equation}

Similar to Section \ref{sec:theory}, we model reranking as a permutation $\sigma \in S_M$ applied to the
ranked list, yielding $L^\sigma_M(Q)$.
The reranking-induced effectiveness change in DCG is given by
\begin{equation}
\Delta_{\mathrm{DCG}}(Q,\sigma)
=
\left|
\mathrm{DCG}(Q; L_M(Q)) - \mathrm{DCG}(Q; L^\sigma_M(Q))
\right|.
\label{eq:del-dcg}
\end{equation}

\subsection{Effect of Rank Perturbations on DCG}

Let $K_i$ denote the random variable corresponding to the rank of the $i$-th relevant
document in the initially retrieved list $L_M(Q)$, with realization $K_i = k_i$, and
let $R_i$ denote the random variable representing its rank after reranking.
The reranking-induced change in DCG for this document, as per the notation of Equation \ref{eq:dcg}, is therefore
\[
\left| f(k_i) - f(r_i) \right|.
\]

To bound this quantity, we exploit the smoothness of the DCG discount function
$f(x) = 1/\log_2(x+1)$.
For $x \ge 1$, $f$ is continuously differentiable and monotonically decreasing.
Applying a first-order Taylor expansion of the DCG discount function
\begin{equation}
f(x) = \frac{1}{\log_2(x+1)},
\label{eq:dcg_discount}
\end{equation}
around $x = k_i$ yields
\begin{equation}
\left| f(R_i) - f(k_i) \right|
\le
\left| f'(k_i) \right| \, |R_i - k_i|
+ \pazocal{O}(|R_i - k_i|^2).
\label{eq:dcg_taylor}
\end{equation}
To compute the derivative appearing in Equation~\ref{eq:dcg_taylor}, we differentiate
the discount function in Equation~\ref{eq:dcg_discount}:
\begin{equation}
f'(x)
=
-\frac{1}{(x+1)\,(\log_2(x+1))^2\,\ln 2}.
\label{eq:dcg_derivative}
\end{equation}
Substituting the derivative into Equation~\ref{eq:dcg_derivative} gives
\begin{equation}
\left| f'(k_i) \right|
=
\frac{1}{(k_i+1)\,(\log_2(k_i+1))^2\,\ln 2}.
\label{eq:dcg_derivative_at_ki}
\end{equation}

For large $k_i$, constant factors and lower-order terms do not affect the dominant
scaling, and Equation~\ref{eq:dcg_derivative_at_ki} therefore simplifies to
\begin{equation}
\left| f'(k_i) \right|
=
\Theta\!\left( \frac{1}{k_i \log^2 k_i} \right).
\label{eq:dcg_derivative_scaling}
\end{equation}
Substituting the scaling behavior of the derivative from
Equation~\ref{eq:dcg_derivative_scaling} into the first-order Taylor bound in
Equation~\ref{eq:dcg_taylor}, we obtain
\begin{equation}
\left| f(R_i) - f(k_i) \right|
\le
\Theta\!\left( \frac{1}{k_i \log^2 k_i} \right)
\, |R_i - k_i|
+
\pazocal{O}(|R_i - k_i|^2).
\label{eq:dcg_taylor_substituted}
\end{equation}

Under the rank perturbation models considered in this work, the expected absolute
rank displacement $\mathbb{E}[|R_i - k_i|]$ is finite, and higher-order terms in
Equation~\ref{eq:dcg_taylor_substituted} contribute only lower-order effects to the
expectation.
Discarding these terms yields the bound
\begin{equation}
\left| f(R_i) - f(k_i) \right|
\le
\Theta\!\left( \frac{1}{k_i \log^2 k_i} \right)
\, |R_i - k_i|.
\label{eq:dcg_taylor_linear_bound}
\end{equation}

Taking expectations with respect to the conditional distribution of $R_i$ conditioned on
$K_i = k_i$ in Equation~\ref{eq:dcg_taylor_linear_bound} yields
\begin{equation}
\mathbb{E}
\left[
\left|
\frac{1}{\log_2(k_i+1)}
-
\frac{1}{\log_2(R_i+1)}
\right|
\right]
\le
\Theta\!\left( \frac{1}{k_i \log^2 k_i} \right)
\mathbb{E}\!\left[ |R_i - k_i| \right],
\label{eq:dcg_taylor_bound}
\end{equation}
thus providing an upper bound on the expected change in the DCG contribution
of the $i$-th relevant document under reranking.

\subsection{Uniform Rank Perturbations} \label{ss:uniform-swap-md}
We first consider uniform rank perturbations, for which we
fix a relevant document with original rank $K_i = k_i$, and let $R_i$ denote its rank after reranking.
Although the post-reranking ranks of different documents are jointly dependent under a permutation, the marginal distribution of $R_i$ is uniform over
$\{1,\ldots,M\} \setminus \{k_i\}$.
That is,
\[
\Pr(R_i = r \mid K_i = k_i) = \frac{1}{M-1}, \quad r \neq k_i.
\]
The expected absolute rank displacement is therefore
\begin{equation}
\mathbb{E}[|R_i - k_i|]
=
\frac{1}{M-1}
\sum_{r \neq k_i} |r - k_i|.    
\label{eq:rank-displacement-md}
\end{equation}
To evaluate this sum, we split it into contributions from ranks below and above
$k_i$:
\[
\sum_{r \neq k_i} |r - k_i|
=
\sum_{r=1}^{k_i-1} (k_i - r)
+
\sum_{r=k_i+1}^{M} (r - k_i).
\]
Each term is a sum of consecutive integers.
Specifically,
\[
\sum_{r=1}^{k_i-1} (k_i - r)
=
\sum_{d=1}^{k_i-1} d
=
\frac{(k_i-1)k_i}{2},
\]
and
\[
\sum_{r=k_i+1}^{M} (r - k_i)
=
\sum_{d=1}^{M-k_i} d
=
\frac{(M-k_i)(M-k_i+1)}{2}.
\]

Substituting these expressions back into the expectation of rank displacements (Equation \ref{eq:rank-displacement-md}) yields
\begin{equation}
\mathbb{E}[|R_i - k_i|]
=
\frac{1}{2(M-1)}
\left(
(k_i-1)k_i
+
(M-k_i)(M-k_i+1)
\right).
\label{eq:md-bound-uniform}
\end{equation}
Clearly, it can seen from Equation \ref{eq:md-bound-uniform} that for any $k_i \in \{1,\ldots,M\}$, the dominant term in the numerator is quadratic
in $M$, while the denominator grows linearly in $M$.
Consequently,
\begin{equation}
\mathbb{E}[|R_i - k_i|]
=
\Theta(M),   \label{eq:bound-md} 
\end{equation}
independent of any specific $k_i$.
Substituting this bound of Equation \ref{eq:bound-md} into Equation~\ref{eq:dcg_taylor_bound} yields
\begin{equation}
\mathbb{E}
\left[
\left|
\frac{1}{\log_2(k_i+1)}
-
\frac{1}{\log_2(R_i+1)}
\right|
\right]
=
\pazocal{O}\!\left( \frac{M}{k_i \log^2 k_i} \right).
\label{eq:big-oh-md}
\end{equation}

Summing these contributions over all relevant documents yields the following result.
\begin{proposition}[Reranking risk for DCG under uniform rank perturbations]
\label{prop:dcg_uniform}
For a ranked list $L_M(Q)$ with relevant document ranks
$K = \{k_1,\ldots,k_m\}$ ($k_1 < k_2 < \ldots < k_m$), the expected reranking-induced effectiveness
change under uniform rank perturbations satisfies
\[
\mathbb{E}\!\left[ \Delta_{\mathrm{DCG}}(Q,\sigma) \right]
=
\pazocal{O}\!\left(
\sum_{i=1}^m \frac{M}{k_i \log^2 k_i}
\right).
\]
\end{proposition}
The bound in Proposition~\ref{prop:dcg_uniform} depends on the individual ranks of
all relevant documents.
While this form is asymptotically tight, it further can be simplified to yield a more interpretable, albeit looser, characterization of reranking risk.

\begin{corollary}[Simplified DCG risk bound under uniform rank perturbations]
\label{cor:dcg_uniform_simple}
Under uniform rank perturbations, the expected reranking-induced change in DCG
satisfies
\[
\mathbb{E}\!\left[ \Delta_{\mathrm{DCG}}(Q,\sigma) \right]
=
\pazocal{O}\!\left(
\frac{m M}{k_1 \log^2 k_1}
\right).
\]
\end{corollary}
\begin{proof}
Note that $k_1 = \min_i k_i$ denotes the highest-ranked relevant document in $L_M(Q)$.
From Proposition~\ref{prop:dcg_uniform}, we have
\[
\mathbb{E}\!\left[ \Delta_{\mathrm{DCG}}(Q,\sigma) \right]
=
\pazocal{O}\!\left(
\sum_{i=1}^m \frac{M}{k_i \log^2 k_i}
\right).
\]
Since the function $x \log^2 x$ is monotonically increasing for $x \ge 2$, it follows
that $k_i \log^2 k_i \ge k_1 \log^2 k_1$ for all $i$.
Substituting this bound yields
\[
\sum_{i=1}^m \frac{M}{k_i \log^2 k_i}
\le
\frac{m M}{k_1 \log^2 k_1},
\]
which proves the claim.
\end{proof}
Corollary~\ref{cor:dcg_uniform_simple} shows that, under uniform rank perturbations,
reranking risk is dominated by the position of the highest-ranked relevant document
and grows linearly with both the ranked list depth and the number of relevant
documents.

\subsection{Locality-Biased Rank Perturbations}
\label{ss:locality-swap-md}

We next consider locality-biased rank perturbations, in which the probability of
displacing a document decays exponentially with the magnitude of the displacement.
As discussed in Section~\ref{ss:localrankswapmodel-sd}, such perturbations represent a
conservative reranking behaviour, where documents are more likely to move only by a small
number of rank positions rather than undergo large displacements.
Consequently, reranking-induced effectiveness changes are expected to be smaller
than under uniform rank perturbations analyzed in
Section~\ref{ss:uniform-swap-md}.

As in the uniform case, we analyze reranking risk by substituting an appropriate
expression for the expected absolute rank displacement
$\mathbb{E}[|R_i - k_i|]$ into the Taylor-based bound of
Equation~\ref{eq:dcg_taylor_bound}.
We therefore begin by deriving this expectation under the locality-biased rank
perturbation model.

Fix a relevant document with original rank $K_i = k_i$, and let $R_i$ denote its rank after reranking. Under the locality-biased model (Equation~\ref{eq:rankswapmodel}), the marginal
conditional distribution of $R_i$ given $K_i = k_i$ is
\begin{equation}
\Pr(R_i = r \mid K_i = k_i)
=
\frac{\exp(-\beta |r-k_i|)}{Z_{k_i}},
\quad r \neq k_i,
\label{eq:md-locality-ranker}
\end{equation}
where the normalization constant is
\[
Z_{k_i} = \sum_{r \neq k_i} \exp(-\beta |r-k_i|).
\]

In Equation \ref{eq:expected-rd-localswap} of Section \ref{ss:localrankswapmodel-sd} we already derived that the expected value of the absolute differences between these random variables is 
\begin{equation}
\mathbb{E}[|R_i - k_i|]
=
\Theta\!\left( \frac{1}{\beta} \right),
\label{eq:beta-intuition}
\end{equation}
for any fixed $\beta > 0$, reflecting the concentration of probability mass near the original rank as locality increases.

Substituting this bound into the Taylor series based inequality of
Equation~\ref{eq:dcg_taylor_bound} yields
\begin{equation}
\mathbb{E}
\left[
\left|
\frac{1}{\log_2(k_i+1)}
-
\frac{1}{\log_2(R_i+1)}
\right|
\right]
=
\pazocal{O}\!\left(
\frac{1}{\beta\, k_i \log^2 k_i}
\right).
\label{eq:dcg_local_single}
\end{equation}

Summing these contributions over all recall points in $L_M(Q)$ yields the following result.
\begin{proposition}[Reranking risk for DCG under locality-biased rank perturbations]
\label{prop:dcg_local}
For a ranked list $L_M(Q)$ with $m$ ($m < M$) recall points
$K = \{k_1,\ldots,k_m\}$, the expected reranking-induced effectiveness
change under locality-biased rank perturbations satisfies
\begin{equation}
\mathbb{E}\!\left[ \Delta_{\mathrm{DCG}}(Q,\sigma) \right]
=
\pazocal{O}\!\left(
\sum_{i=1}^m \frac{1}{\beta\, k_i \log^2 k_i}
\right).
\label{eq:final-bound}
\end{equation}
\end{proposition}

\begin{corollary}[Simplified DCG risk bound under locality-biased rank perturbations]
\label{cor:dcg_local_simple}
Again similar to Corollary \ref{cor:dcg_uniform_simple}, as $k_1 = \min_i k_i$ denotes the highest-ranked relevant document, using the same argument on the monotonically increasing property of $x \log^2 x$, 
a simpler (but less strict) bound on the
expected reranking-induced DCG change can thus be written as
\[
\mathbb{E}\!\left[ \Delta_{\mathrm{DCG}}(Q,\sigma) \right]
=
\pazocal{O}\!\left(
\frac{m}{\beta\, k_1 \log^2 k_1}
\right).
\]
\end{corollary}

An interesting observation is that the uniform and locality-biased DCG bounds (Propositions \ref{prop:dcg_uniform} and \ref{prop:dcg_local}) differ
only in that the ranked list depth $M$ is replaced with the inverse locality parameter
$1/\beta$, i.e., setting $\beta = 1/M$ in Corollary~\ref{cor:dcg_local_simple} leads to the uniform bound of Corollary~\ref{cor:dcg_uniform_simple}.
From this perspective, the locality parameter $\beta$ can be interpreted as
controlling the scale of rank uncertainty.

\section{Experiment Setup}
\label{sec:experimental-setup}


In this section, we describe the experiment setup to empirically investigate the extent to which the behaviour of
stochastic reranking methods in practice conforms to the theoretical predictions
developed in Sections \ref{sec:theory} and \ref{sec:dcg_multi}.

While the theory provides insights into how reranking risk scales under idealised
stochastic perturbations with a closed-form underlying probability distribution for the ranking policy.
However, real-world stochastic rankers involve multiple sources of randomness, such as
$\epsilon$-decay or exploration--exploitation mechanisms
\cite{Feng2022HasCG,Yadav2019PolicyGradientTO,Oh2023ApplyingRL} or information from sensitive attributes (e.g., race or gender) from document metadata  \cite{sen,jaenich2022university}, the probability distributions of which are difficult to express as analytic functions.
The research question is therefore to 
assess whether the observed effectiveness changes in real-life stochastic rankers follow the trends and scaling behaviour anticipated by the theoretical analysis.


\subsection{Dataset and Stochastic Ranked Lists} \label{ss:dataset}
We conduct our experiments on the `Task 2' (Wikipedia Editors) of the TREC Fair ranking track \cite{DBLP:conf/trec/EkstrandMR022}, which involves retrieving 100 ranked lists for each query with an objective to optimize fairness across gender and geographic location. 

For our experiments, we leverage the three best performing runs of the TREC Fair Ranking task \cite{DBLP:conf/trec/EkstrandMR022} -- all employing a multi-armed bandit (MAB) based optimization policy.
Each run in the TREC Fairness 2022 task constitutes a total of 100 submitted result lists, each with a sequence number (representing time). The first list is obtained by a deterministic model (specifically, ColBERT \cite{colbert} for the runs that we experiment with \cite{DBLP:conf/trec/JanichMO22}). Each subsequent list reflects a stochastic re-ranking of the initially retrieved list of candidate documents.
More details about the runs used in our experiments are as follows.
\uls
\item \textbf{MAB} (official submission name ``UoGTrMabSAED''):
A stochastic MAB-based approach that incorporates randomization in the bandit component but does not use sensitive-attribute-based weighting. It uses an $\epsilon$-greedy strategy of exploration-exploitation trade-off with a fixed value of $\epsilon$.

\item \textbf{MAB-ED} (official submission name ``UoGTrMabSaWR''): Similar to \textbf{MAB}, the only difference is that it has a decaying value of $\epsilon$, i.e., the ranking policy becomes less stochastic with time (iteration number).

\item \textbf{MAB-SA} (official submission name ``UoGTrMabWeSA''): Similar to MAB, but this method also uses the fairness weights derived from sensitive attribute values to sample a rank ordering.  
\ule

\subsection{Prediction Task and Evaluation Measure}
Let us denote a ranking policy (i.e., one of the three approaches listed in Section \ref{ss:dataset}) by $\pi$.
Each submitted run in the TREC fairness task for a query $Q$ is a sequence of ranked lists (each of size $M=20$) of the form:
$$\pazocal{S}(Q,\pi) = \langle L^{(0)}(Q,\pi),\, L^{(1)}(Q,\pi),\, \ldots\, L^{(n-1)}(Q,\pi)\rangle$$ where $n=100$. To avoid clutter we omit the suffix $M$ denoting the size of a ranked list, as per the notations developed in Equation \ref{eq:delta_mu_def}. Each $L^{(i)}(Q,\pi) \sim \pi(Q, L^{(0)}(Q))$ for $i=1,\ldots,n-1$, i.e., the first ranked list in the sequence is the initially retrieved list while the other ones are sampled from the stochastic policy $\pi$. 



As a next step, we compute the actual (ground-truth) DCG change induced by the $i$-th sample by substituting: $L^\sigma_M(Q) \gets L^{(i)}(Q,\pi)$ in
Equation \ref{eq:del-dcg} for each $i=1,\ldots,n-1$, which for convenience we denote as $\Delta_{\mathrm{DCG}}(Q, \pi, i)$. This step uses the TREC Fairness track 2022 relevance assessments.

Each query-ranker pair thus leads to an observation in the actual DCG change of the form $\Delta_{\mathrm{DCG}}(Q, \pi, i)$. The objective is now to predict an upper bound of this reranking risk $\hat{\Delta}_{\mathrm{DCG}}(Q, \pi, i)$
as a function of $\sum_{j=1}^m 1/(\beta\, k_j \log^2 k_j)$, i.e., by applying Equation \ref{eq:final-bound}.


The bound of Equation~\ref{eq:final-bound} depends on a ranker-specific locality parameter $\beta$, which controls the decay rate of rank displacement probabilities under the stochastic perturbation model in Equation~\ref{eq:md-locality-ranker}. For each stochastic ranker $\pi$, we estimate $\beta$ empirically from observed rank displacements of the first relevant document. A natural deployment-oriented setting is one in which only an initial fraction of stochastic ranking realisations is available for estimating $\beta$, after which the estimated value is used to predict the risk associated with subsequent, unobserved realisations.
Formally,
\begin{equation}
\begin{split}
\pazocal{T}(Q,\pi) &= \bigcup_{i=1}^{\tau \cdot n} \{L^{(i)}(Q,\pi)\}, \\
\beta(Q, \pi) &= \phi(\pazocal{T}(Q,\pi)),
\end{split}
\end{equation}
where $\phi$ denotes the estimation procedure for $\beta$ from the training set $\pazocal{T}(Q,\pi)$ of observed rank samples for query $Q$, and
$\tau \in [0,1]$ is the proportion of rank samples assumed to be observed.

For our experiments, we set $\tau=0.1$, corresponding to 10 (out of 100) observed rank samples per query-ranker pair.
A small value of $\tau=0.1$ reflects the intended use of the proposed framework as an early-stage risk
estimation procedure where the locality parameter is inferred from limited evidence, and then used to estimate risk for subsequent ranking outputs.
Although larger values of $\tau$ would likely produce more stable estimates of $\beta$, they would also make the experiment closer to retrospective analysis than to ex ante risk prediction.


Recall from Equation \ref{eq:beta-intuition} that the expected value of the rank shifts is $\Theta(1/\beta)$, which means that the average shifts of the ranks of the relevant documents in the observed stochastic lists relative to the initially retrieved one can be used to estimate $\beta$. More formally,
\begin{equation}
\begin{split}
\bar{D}(\pazocal{T}(Q,\pi)) &= \sum_{L \in \pazocal{T}(Q,\pi)} \sum_{D \in \pazocal{R}(Q)} \left| \rho(D,L) - \rho(D,L^{(0)}(Q,\pi)) \right|, \\    
\beta(Q, \pi) & = \frac{\alpha}{\bar{D}(\pazocal{T}(Q,\pi))},
\label{eq:beta-estimate}
\end{split}
\end{equation}
where $\rho(D,L)$ denotes the rank of a document $D$ in a list $L$, and $\alpha \in \mathbb{R}$ is a constant for the asymptotic scaling. In our experiments, we found that a value of $\alpha \in [0.2, 0.5]$ works well in practice.

Using the theoretical bound from Equation \ref{eq:final-bound}, for each unobserved ranked list of each query, i.e., for each $L^{(j)}(Q,\pi) \in \pazocal{S}(Q,\pi) - \pazocal{T}(Q,\pi)$, we predict
\begin{equation}
\hat{\Delta}_{\mathrm{DCG}}(Q, \pi, j) = \frac{1}{\beta(Q,\pi)} \sum_{D \in \pazocal{R}(Q)} \frac{1}{\rho(D,L) \log^2(\rho(D,L))},
\end{equation}
where we use the value of $\beta$ from Equation \ref{eq:beta-estimate} estimated on the training set of observed rank permutations for a particular policy for a particular query.
Code for the reported experiments is available at \url{https://github.com/gdebasis/srea}.

\begin{figure*}[t]
  \centering
  \begin{subfigure}[t]{0.26\textwidth}
    \includegraphics[width=\linewidth]{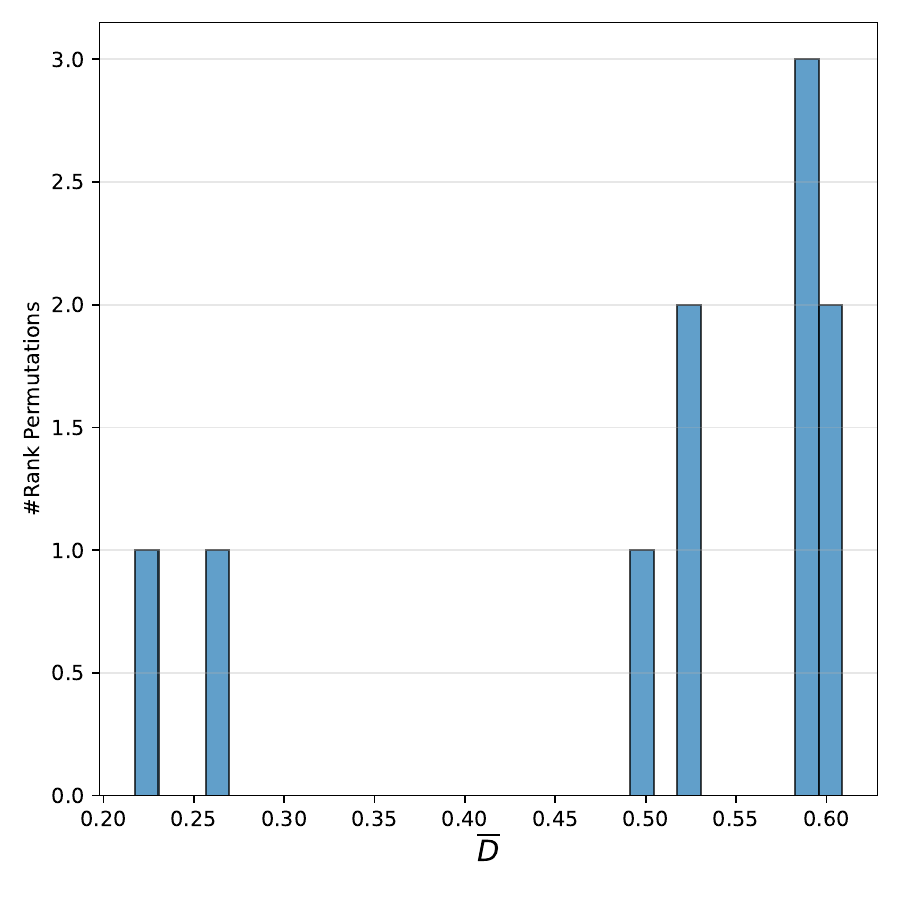}
    \caption{Mean displacement $\bar{D}$}
  \end{subfigure}
  \begin{subfigure}[t]{0.34\textwidth}
    \includegraphics[width=\linewidth]{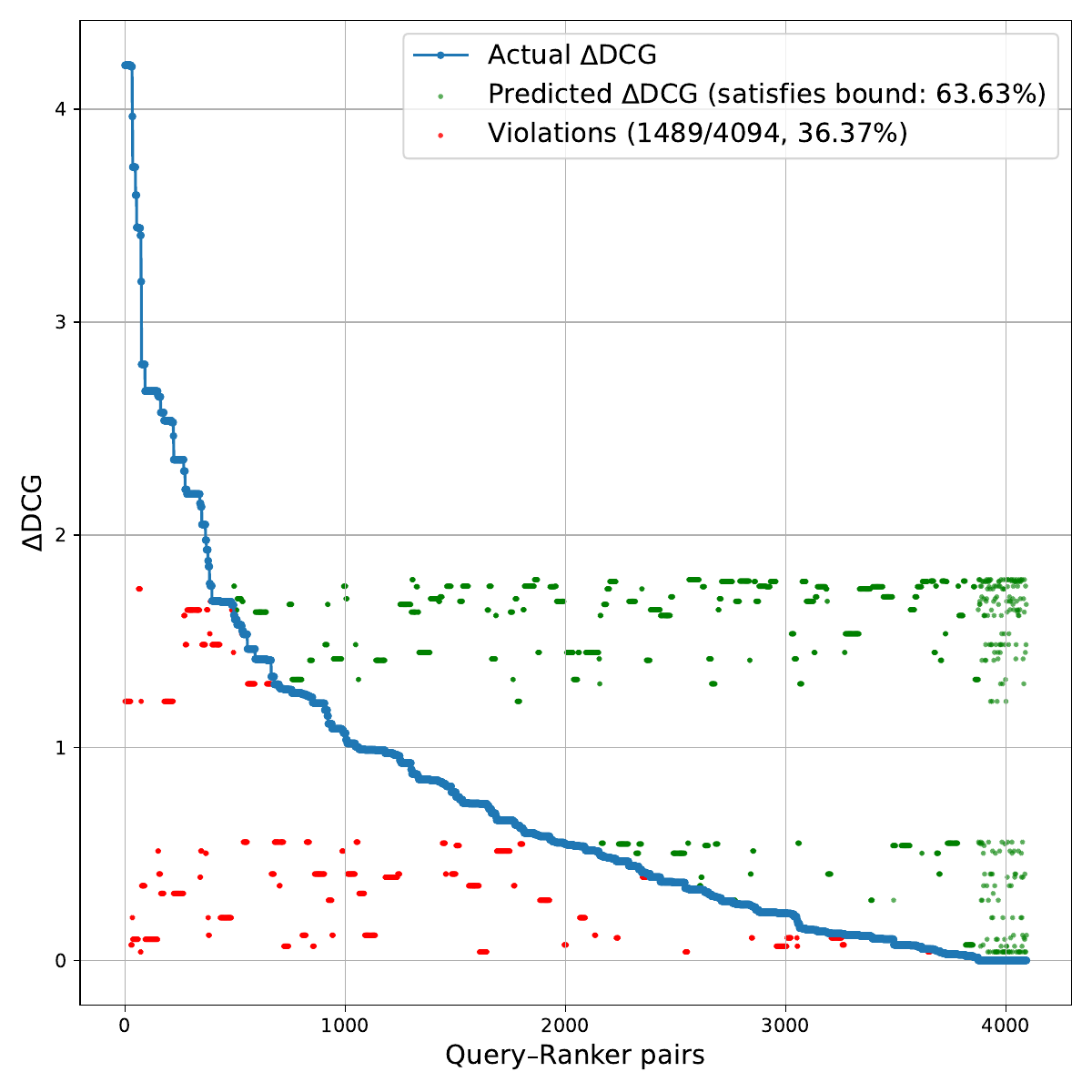}
    \caption{$\alpha=0.5$}
  \end{subfigure}
  \begin{subfigure}[t]{0.34\textwidth}
    \includegraphics[width=\linewidth]{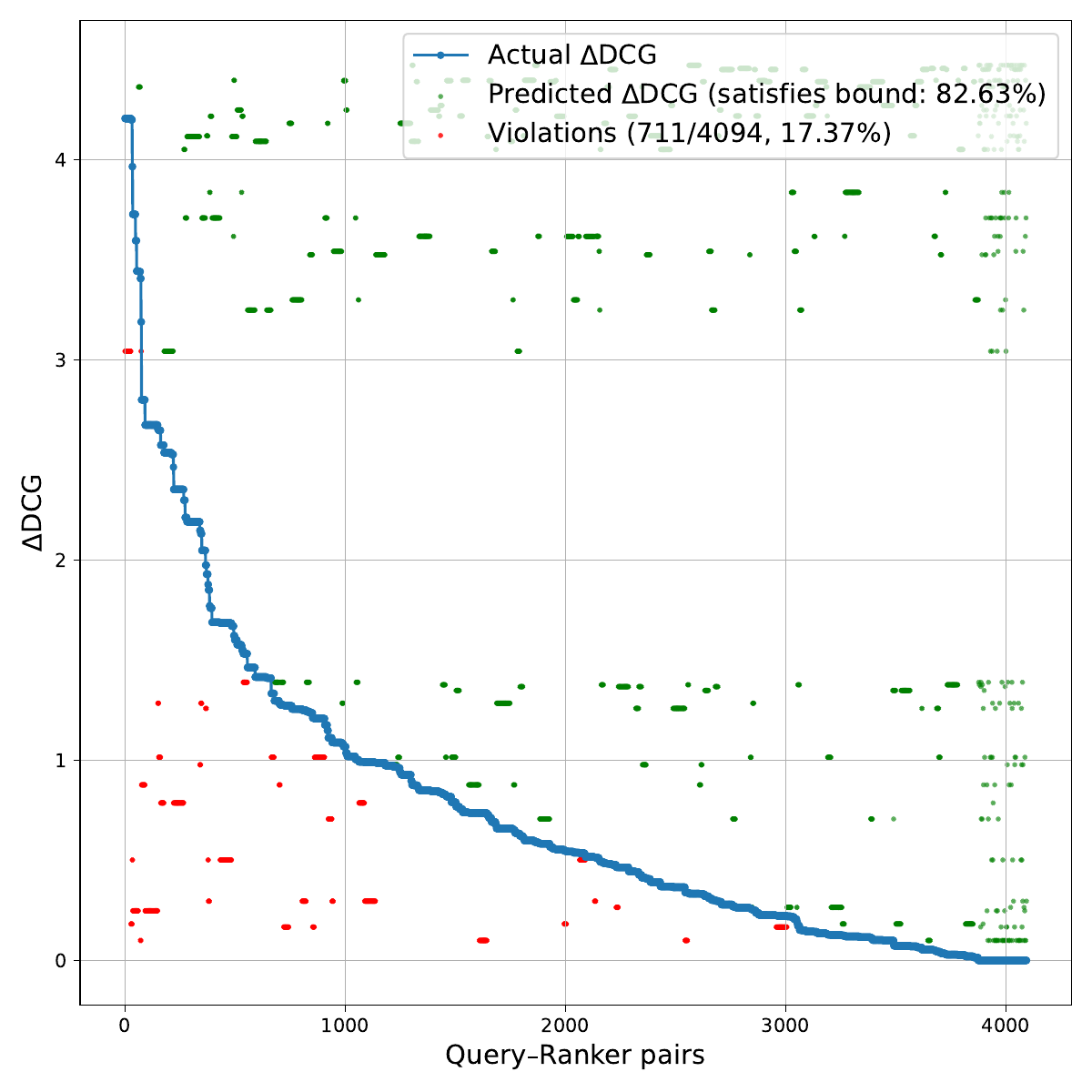}
    \caption{$\alpha=0.2$}
  \end{subfigure}

  \caption{%
  The left-most plot shows the distribution of rank displacements for 10 permutation samples (considered as observed) averaged over 48 queries of TREC Fairness 2022 dataset. The rest of the plots show the accuracy of the bounds, as predicted by Equation \ref{eq:final-bound}, for unobserved permutations over each query ($48\times(99-10)$ such pairs) for the ranking policy: MAB (Multi-Arm Bandits with fixed stochasticity with no weights from sensitive attributes) for two different $\alpha$ values.
  }
  \label{fig:mab-sawr}
\end{figure*}
\begin{figure*}[t]
  \centering
  \begin{subfigure}[t]{0.26\textwidth}
    \centering
    \includegraphics[width=\linewidth]{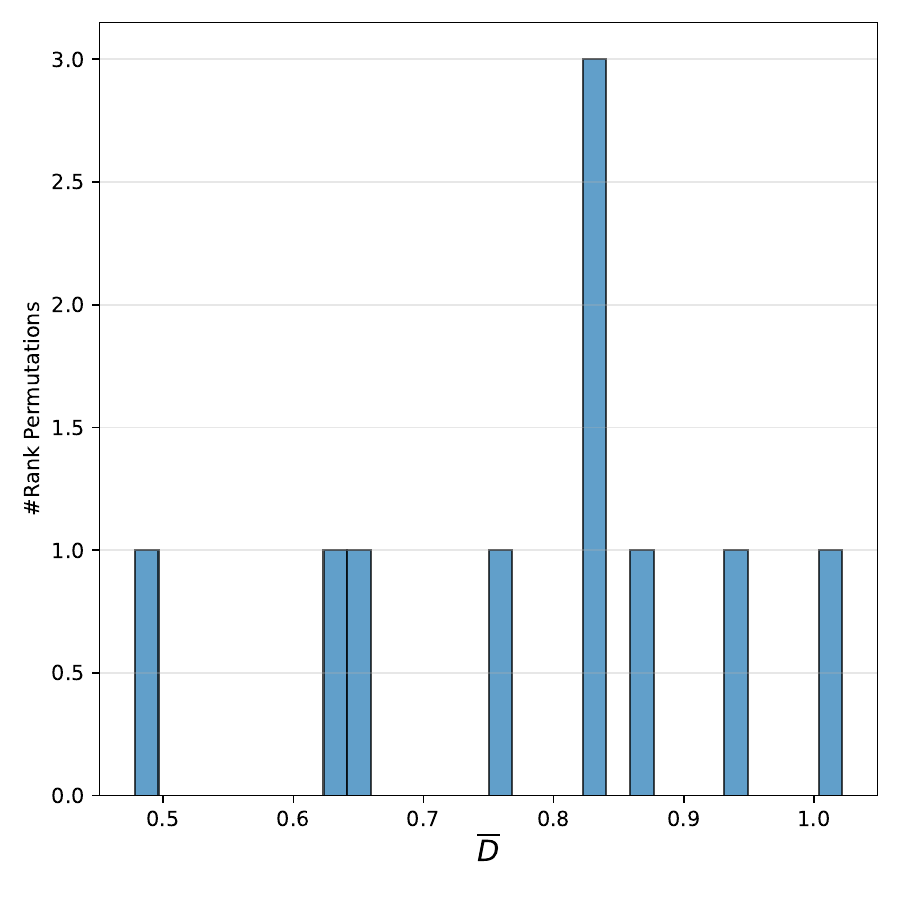}
    \caption{Mean displacement $\bar{D}$}
  \end{subfigure}
  \begin{subfigure}[t]{0.34\textwidth}
    \includegraphics[width=\linewidth]{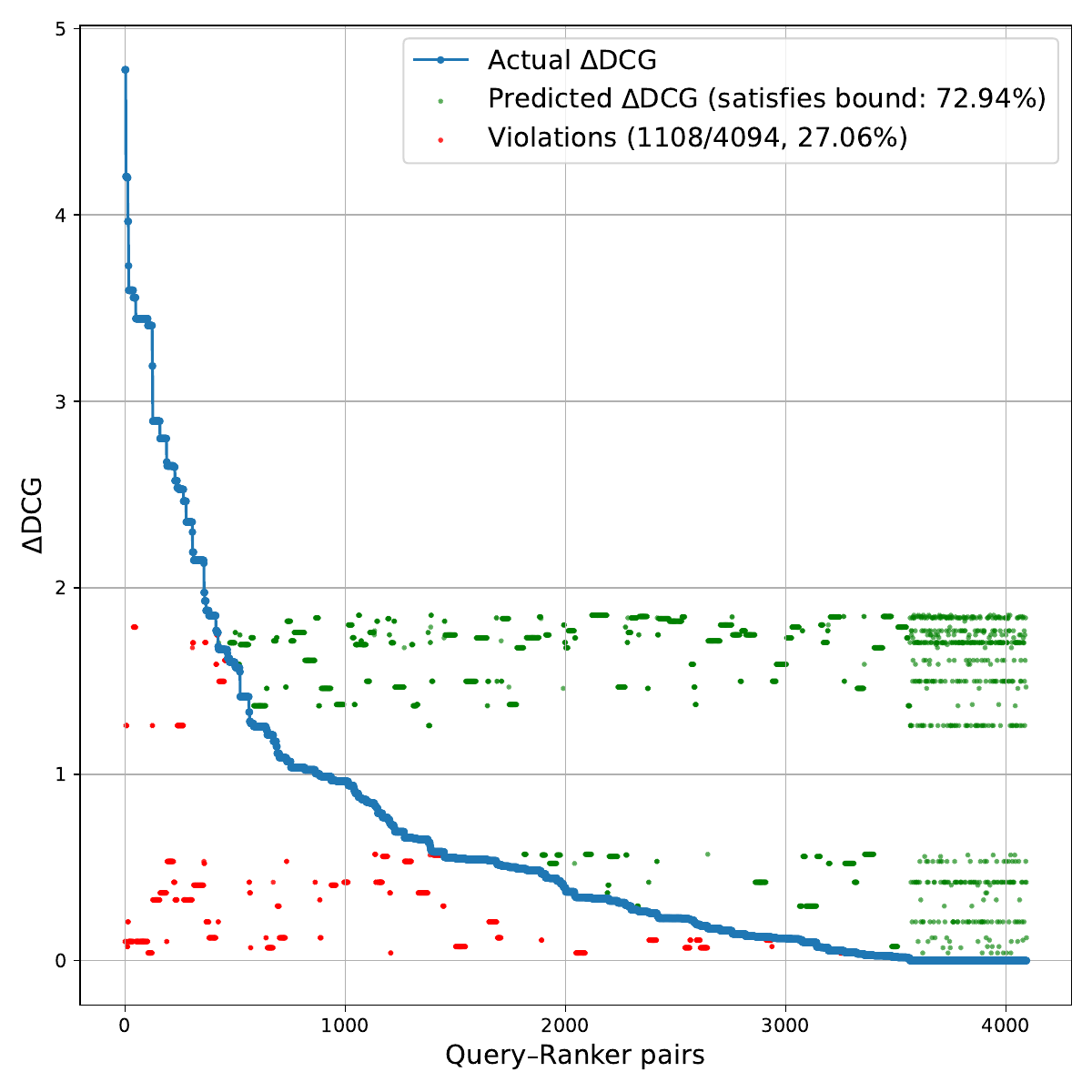}
    \caption{$\alpha=0.5$}
  \end{subfigure}
  \begin{subfigure}[t]{0.34\textwidth}
    \includegraphics[width=\linewidth]{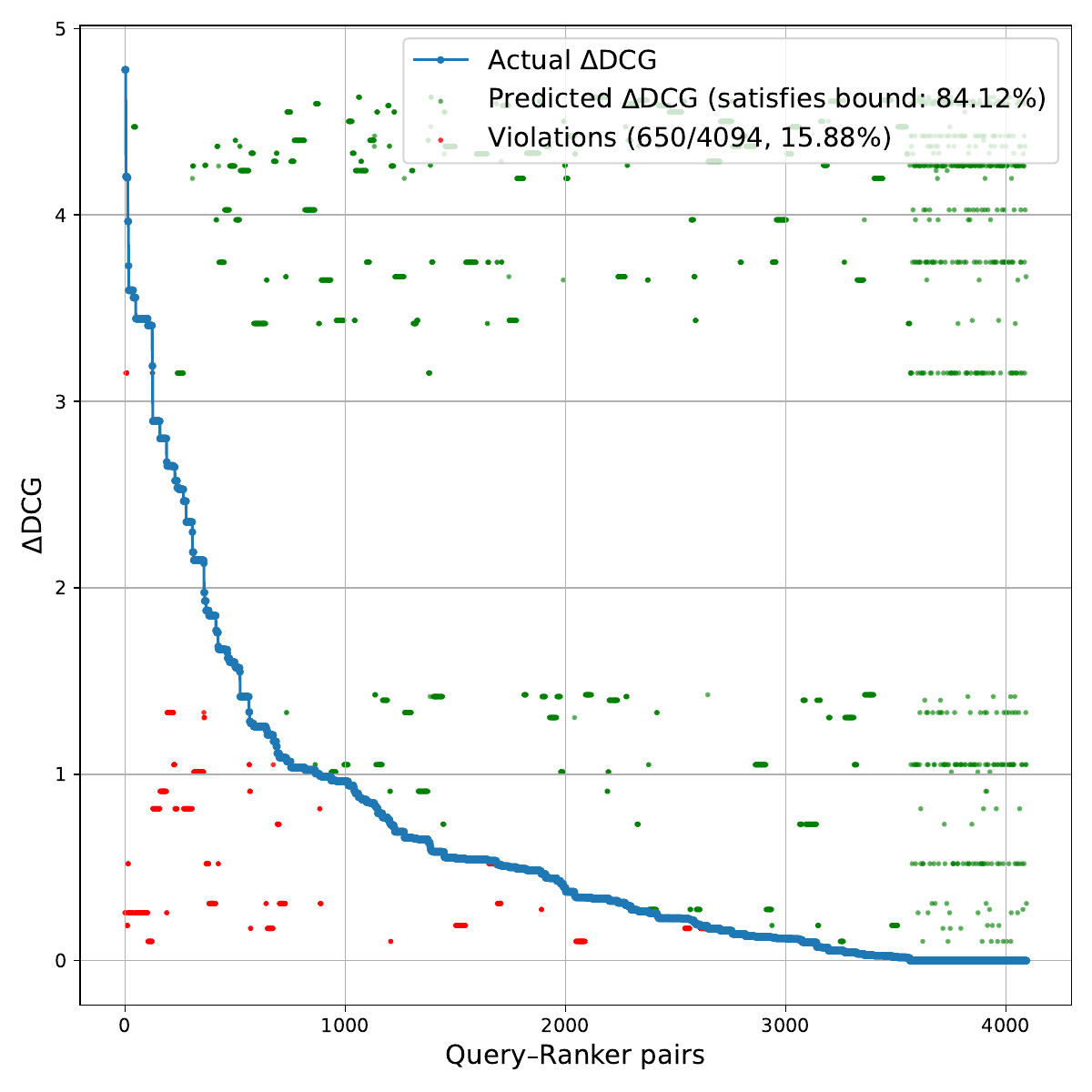}
    \caption{$\alpha=0.2$}
  \end{subfigure}
  \caption{%
  Results with an identical experiment setup to that of Figure \ref{fig:mab-sawr} reported for
  the ranking policy: MAB-ED (Multi-Arm Bandit with an $\epsilon$-greedy exploration strategy with a decaying $\epsilon$ without any sensitive attribute specific weights).
  }
  \label{fig:mab-saed}
\end{figure*}
\begin{figure*}[t]
  \centering
  \begin{subfigure}[t]{0.26\textwidth}
    \includegraphics[width=\linewidth]{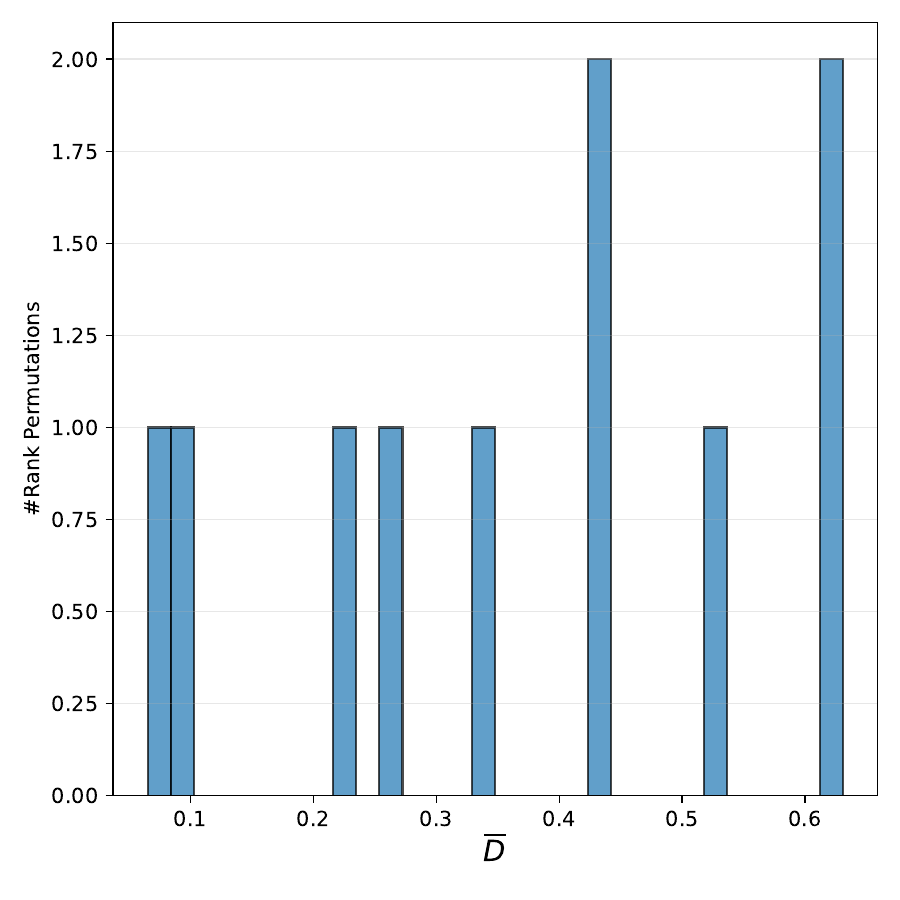}
    \caption{Mean displacement $\bar{D}$}
  \end{subfigure}
  \begin{subfigure}[t]{0.34\textwidth}
    \includegraphics[width=\linewidth]{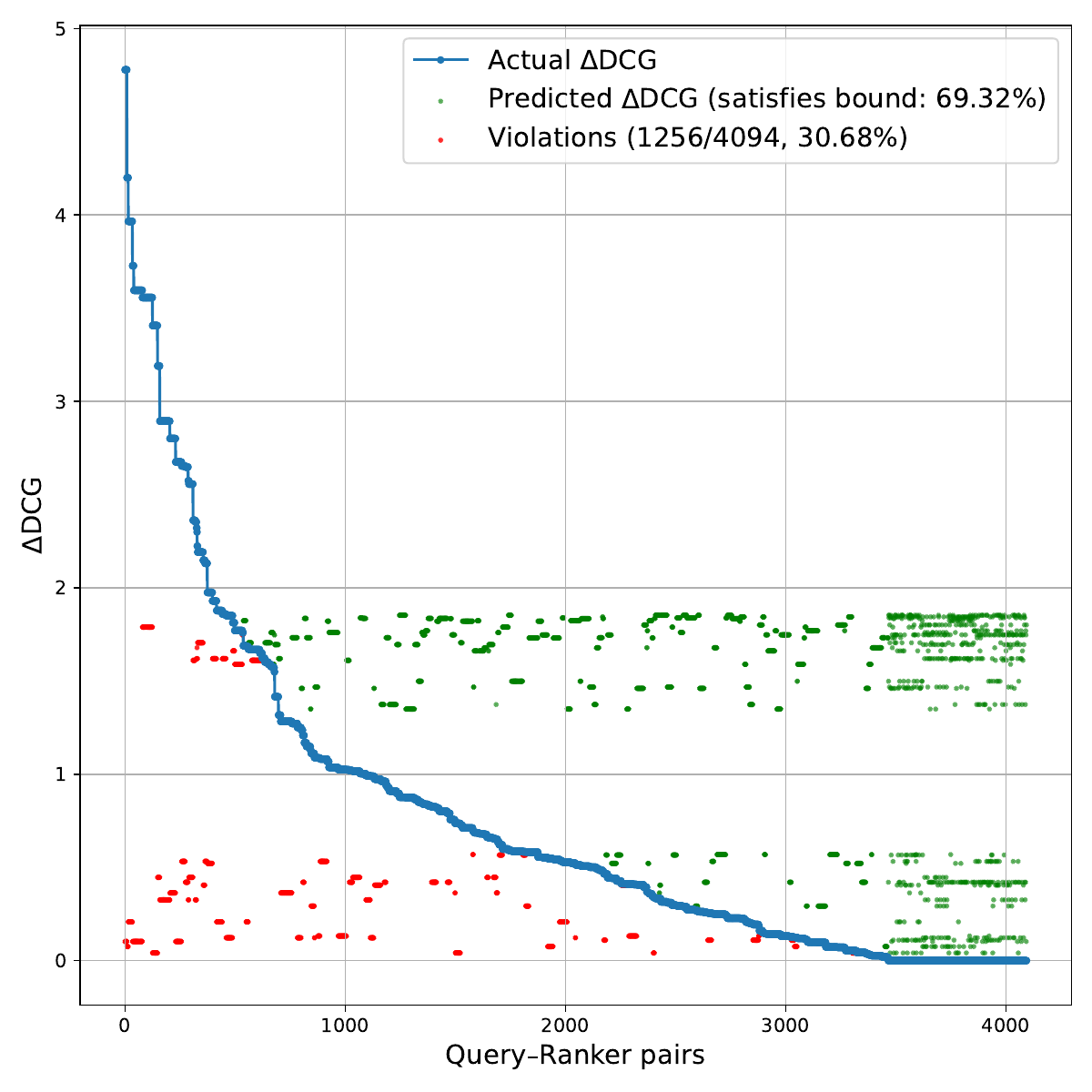}
    \caption{$\alpha=0.5$}
  \end{subfigure}
  \begin{subfigure}[t]{0.34\textwidth}
    \includegraphics[width=\linewidth]{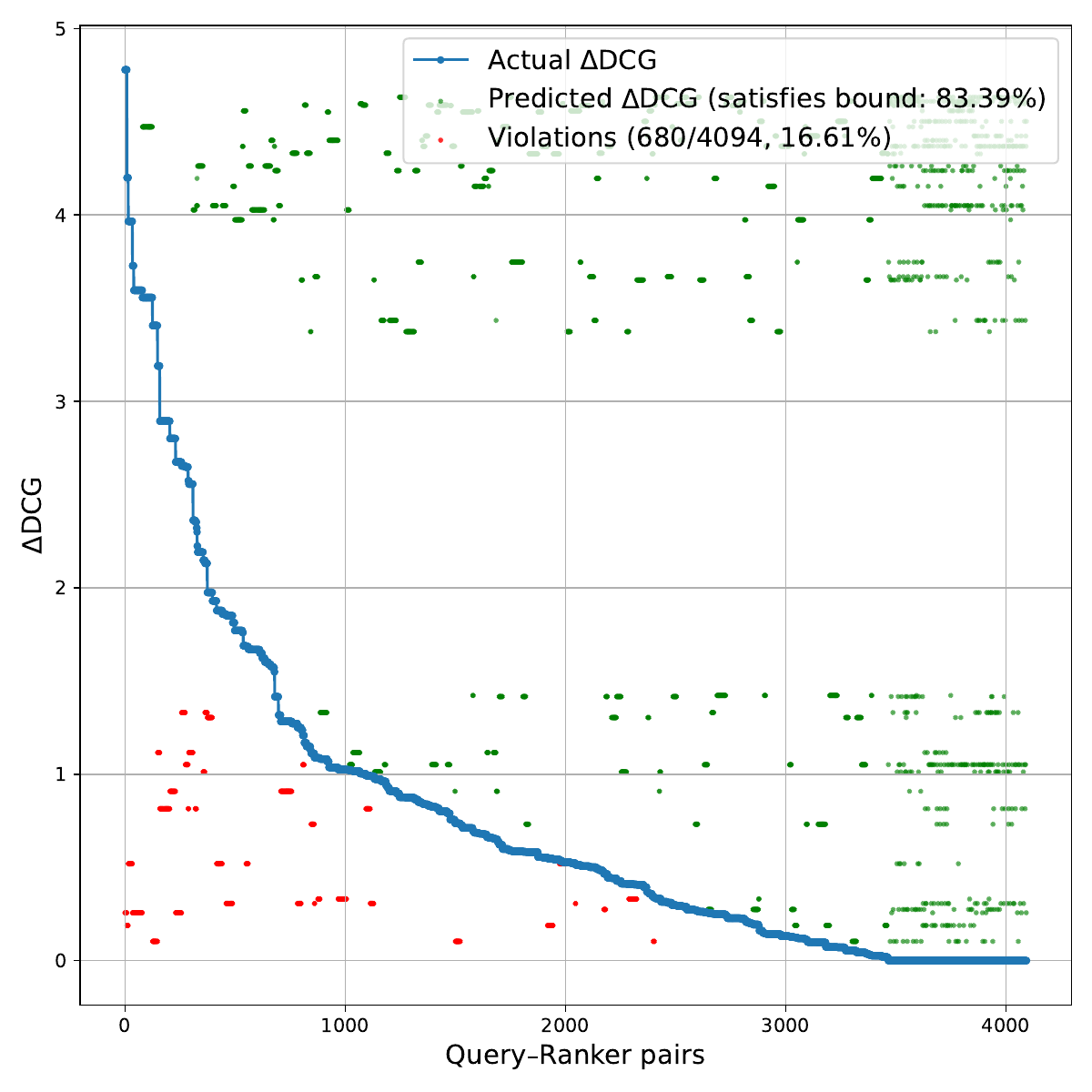}
    \caption{$\alpha=0.2$}
  \end{subfigure}

  \caption{%
Results with an identical experiment setup to that of Figure \ref{fig:mab-sawr} reported for
  the ranking policy:   
  MAB-SA (Multi-Arm Bandit with $\epsilon$-greedy exploration strategy together with fairness weights derived from sensitive attribute values).
  }
  \label{fig:mab-wesa}
\end{figure*}

\section{Results}
\label{sec:results}

As shown in Figures~\ref{fig:mab-sawr}--\ref{fig:mab-wesa},
the different MAB-based stochastic ranking policies induce distinct
distributions of mean rank displacements, reflecting varying degrees of
locality bias arising from their respective exploration strategies.
These displacement patterns provide a first indication of how closely
each policy aligns with the locality assumptions underlying the
theoretical analysis.

Figures~\ref{fig:mab-sawr}--\ref{fig:mab-wesa} also present a pointwise
comparison between the empirically observed DCG changes and the
theoretically predicted DCG sensitivities at the level of individual
query--ranker realizations.
Within each figure, realizations are sorted in decreasing order of
observed DCG change.
Green points indicate cases in which the theoretical prediction
dominates the observed change, while red points denote realizations
that violate the bound.

We observe that the theoretical bound is not satisfied pointwise for all
realizations.
This is not unexpected, since practical ranking policies only
approximately behave like the idealized locality-biased rank
perturbation model assumed in Equation~\ref{eq:md-locality-ranker}.
Moreover, the parameters of the approximated distribution are only
estimated from a finite number of rank samples.
Recall that in our experimental setting, these parameters are estimated from only
10 samples per query and ranking policy, which introduces additional
sampling variability, particularly for realizations involving large
rank displacements.

Among the evaluated methods, the $\epsilon$-decay based policy
(MAB-ED; Figure~\ref{fig:mab-saed}) exhibits the strongest agreement
with the theoretical risk estimates.
A likely explanation is that this policy becomes progressively less
stochastic over time, converging toward a largely deterministic ranking
driven by query--document relevance signals.
As a result, its induced rank perturbations are more tightly
concentrated and more closely matched by the locality-based model.

The policy incorporating sensitive attribute information (MAB-SA) also
demonstrates consistent alignment with the predicted sensitivities,
despite its reliance on proportional exposure constraints over multiple
attribute combinations.
This suggests that the theoretical framework remains informative even
in the presence of complex, attribute-aware randomization mechanisms.

A consistent pattern across all settings is that bound violations occur
predominantly in the \emph{head} of the distribution, corresponding to
realizations with large observed DCG changes.
These cases typically arise when early relevant documents undergo substantial rank displacement, where small estimation errors in the perturbation parameters may translate into larger prediction
inaccuracies.

In contrast, the \emph{tail} of the distribution---constituting the vast
majority of query-ranker realizations with small DCG changes---is largely
well covered by the theoretical predictions.
Across all values of $\alpha$, the predicted sensitivities dominate the
observed DCG changes for most low-variation realizations, indicating
that the theory accurately captures typical, rather than extreme,
reranking behaviour.

Finally, the left plots of Figures~\ref{fig:mab-sawr}--\ref{fig:mab-wesa} illustrates the effect of
the parameter $\alpha$ on the tightness of the predictions.
Decreasing $\alpha$ yields numerically closer alignment with the
observed DCG changes, at the expense of a larger number of bound
violations.
Conversely, larger values of $\alpha$ produce more conservative
predictions that dominate the observed changes for a greater fraction
of realizations.
This trade-off highlights the role of $\alpha$ as a tunable parameter
governing the balance between tightness and robustness of the risk
estimates.

\section{Conclusions and Future Work}
\label{sec:conclusions}

In this paper, we presented a theoretical framework for analysing the
\emph{reranking risk} induced by stochastic ranking policies.
Unlike prior work that evaluates stochastic rankers primarily through expected
effectiveness or exposure after reranking, our approach addresses an ex ante
question: before applying stochastic reranking, how large can the induced
variation in rank-based effectiveness be?
Our theoretical analysis demonstrates that reranking risk is governed by two
fundamental structural factors: the positions of relevant documents in the
initial ranking, and the degree of rank movement permitted by the stochastic
policy.
For the single relevant document case, we derived exact expressions and
asymptotic bounds for the expected change in reciprocal rank under both uniform
and locality-biased rank perturbations.
We extended the analysis to multiple relevant documents using DCG, yielding
additive bounds that decompose reranking risk into contributions from individual
recall points.

We empirically validated the theory using stochastic ranking runs submitted to
the TREC Fairness 2022 track.
Despite the simplifying assumptions required for analytical tractability, the
observed effectiveness changes closely followed the qualitative trends and
scaling behaviour predicted by the theory, particularly for the majority of
low-variation query--ranker realizations.

While the proposed framework provides a principled analytical
characterization of reranking risk, it is not intended as a prescriptive
reranking algorithm or an optimisation strategy in its current form.
Rather, it offers an approximate sensitivity analysis based on idealised
perturbation models that capture key structural properties of stochastic
ranking.
Consequently, the theoretical predictions are not expected to yield tight
pointwise guarantees for all query--ranker realisations, especially under
practical policies that incorporate complex sources of randomness and
additional constraints.
This is also reflected in the empirical observations: although the predicted changes in nDCG generally align well with the observed effectiveness changes, deviations do occur.
We therefore view the framework as a useful analytical abstraction for
understanding typical reranking behaviour and for guiding intuition, rather than as a strict theoretic guarantee.

The notion of reranking risk opens several promising directions for future work.
One avenue is the development of \emph{risk-aware} stochastic ranking policies
that explicitly balance relevance, fairness, and effectiveness stability by
modulating their degree of randomisation.
Another direction is to integrate reranking risk estimation with user
interaction models, for example by leveraging click data or inferred relevance
signals to obtain query-specific risk control in online settings.
More broadly, bridging the gap between theoretical characterisation and
actionable reranking strategies---such as incorporating risk estimates into optimisation objectives or decision policies---remains an important direction for
future work.

\section*{Acknowledgements}
The author thanks the anonymous reviewers for their careful reading and constructive
comments. Their feedback helped improve the clarity and correctness of the theoretical
analysis.

\appendix

\section{Appendix}

\subsection{Asymptotic Behaviour $S(\beta) = \sum_{d=1}^{\infty} e^{-\beta d}$}
\label{app:normalization}

\begin{equation}
S(\beta)
=
2\,\frac{e^{-\beta}}{1 - e^{-\beta}}.
\label{eq:zk_closed_form}
\end{equation}
To study the asymptotic behaviour of \eqref{eq:zk_closed_form}, we expand the exponential
function for small $\beta$:
\[
e^{-\beta}
=
1 - \beta + \frac{\beta^2}{2} + \pazocal{O}(\beta^3).
\]
It follows that
\[
1 - e^{-\beta}
=
\beta + \pazocal{O}(\beta^2).
\]
Substituting these expressions into \eqref{eq:zk_closed_form} yields
\[
S(\beta)=
2\,\frac{1 - \beta + \pazocal{O}(\beta^2)}{\beta + \pazocal{O}(\beta^2)}
=
\frac{2}{\beta}\left( 1 + \pazocal{O}(\beta) \right) = \Theta\!\left( \frac{1}{\beta} \right),
\quad \text{as } \beta \to 0^+.
\]

\subsection{Asymptotic Behaviour of $S(\beta)=\sum_{d=1}^{\infty} d e^{-\beta d}$}
\label{app:geom_sum}

We analyse the asymptotic behaviour of the series
\[
S(\beta) \;=\; \sum_{d=1}^{\infty} d \, e^{-\beta d},
\qquad \beta > 0,
\]
as $\beta \to 0^+$.
Using the standard identity for geometric series derivatives,
\[
\sum_{d=1}^{\infty} d x^d
=
\frac{x}{(1-x)^2},
\qquad |x| < 1,
\]
and substituting $x = e^{-\beta}$, we obtain the closed-form expression
\begin{equation}
S(\beta)
=
\frac{e^{-\beta}}{(1 - e^{-\beta})^2}.
\label{eq:closed_form}
\end{equation}
To study the asymptotic behaviour of \eqref{eq:closed_form} as $\beta \to 0^+$,
we expand the exponential function using its Taylor series:
\[
e^{-\beta}
=
1 - \beta + \frac{\beta^2}{2} + \pazocal{O}(\beta^3).
\]
It follows that
\[
1 - e^{-\beta}
=
\beta - \frac{\beta^2}{2} + \pazocal{O}(\beta^3),
\]
and therefore
\[
(1 - e^{-\beta})^2
=
\beta^2 \left( 1 - \beta + \pazocal{O}(\beta^2) \right).
\]

Substituting these expansions into \eqref{eq:closed_form}, we have
\[
S(\beta)
=
\frac{1 - \beta + \pazocal{O}(\beta^2)}{\beta^2 \left( 1 - \beta + \pazocal{O}(\beta^2) \right)}
=
\frac{1}{\beta^2} \left( 1 + \pazocal{O}(\beta) \right),
\quad \text{as } \beta \to 0^+.
\]

Hence, we conclude that
\begin{equation}
\sum_{d=1}^{\infty} d \, e^{-\beta d}
=
\Theta\!\left( \frac{1}{\beta^2} \right),
\qquad \beta \to 0^+.
\end{equation}

\bibliographystyle{ACM-Reference-Format}
\balance
\bibliography{refs}

@inproceedings{DBLP:conf/trec/JanichMO22,
  author       = {Thomas J{\"{a}}nich and
                  Graham McDonald and
                  Iadh Ounis},
  title        = {University of Glasgow Terrier Team at the {TREC} 2022 Fair Ranking
                  Track},
  booktitle    = {{TREC}},
  series       = {{NIST} Special Publication},
  volume       = {500-338},
  publisher    = {National Institute of Standards and Technology {(NIST)}},
  year         = {2022}
}

@article{Wu2022JointME,
  title={Joint Multisided Exposure Fairness for Recommendation},
  author={Haolun Wu and Bhaskar Mitra and Chen Ma and Fernando Diaz and Xue Liu},
  journal={Proceedings of the 45th International ACM SIGIR Conference on Research and Development in Information Retrieval},
  year={2022},
  url={https://api.semanticscholar.org/CorpusID:248495930}
}

@article{Togashi2022FairMF,
  title={Fair Matrix Factorisation for Large-Scale Recommender Systems},
  author={Riku Togashi and Kenshi Abe},
  journal={ArXiv},
  year={2022},
  volume={abs/2209.04394},
  url={https://api.semanticscholar.org/CorpusID:252185173}
}

@article{Oosterhuis2021ComputationallyEO,
  title={Computationally Efficient Optimization of Plackett-Luce Ranking Models for Relevance and Fairness},
  author={Harrie Oosterhuis},
  journal={Proceedings of the 44th International ACM SIGIR Conference on Research and Development in Information Retrieval},
  year={2021},
  url={https://api.semanticscholar.org/CorpusID:233481799}
}

@inproceedings{Guo2023InferencetimeSR,
  title={Inference-time Stochastic Ranking with Risk Control},
  author={Ruocheng Guo and Jean-François Ton and Yang Liu},
  year={2023},
  url={https://api.semanticscholar.org/CorpusID:259137461}
}

@article{Vardasbi2022ProbabilisticPG,
  title={Probabilistic Permutation Graph Search: Black-Box Optimization for Fairness in Ranking},
  author={Ali Vardasbi and Fatemeh Sarvi and M. de Rijke},
  journal={Proceedings of the 45th International ACM SIGIR Conference on Research and Development in Information Retrieval},
  year={2022},
  url={https://api.semanticscholar.org/CorpusID:248476139}
}

@article{Bruch2020AST,
  title={A Stochastic Treatment of Learning to Rank Scoring Functions},
  author={Sebastian Bruch and Shuguang Han and Michael Bendersky and Marc Najork},
  journal={Proceedings of the 13th International Conference on Web Search and Data Mining},
  year={2020},
  url={https://api.semanticscholar.org/CorpusID:209388063}
}

@article{Gorantla2023OptimizingGP,
  title={Optimizing Group-Fair Plackett-Luce Ranking Models for Relevance and Ex-Post Fairness},
  author={Sruthi Gorantla and Eshaan Bhansali and Amit Deshpande and Anand Louis},
  journal={ArXiv},
  year={2023},
  volume={abs/2308.13242},
  url={https://api.semanticscholar.org/CorpusID:261214757}
}

@inproceedings{diazcikm2020,
author = {Diaz, Fernando and Mitra, Bhaskar and Ekstrand, Michael D. and Biega, Asia J. and Carterette, Ben},
title = {Evaluating Stochastic Rankings with Expected Exposure},
year = {2020},
isbn = {9781450368599},
publisher = {Association for Computing Machinery},
address = {New York, NY, USA},
url = {https://doi.org/10.1145/3340531.3411962},
doi = {10.1145/3340531.3411962},
booktitle = {Proceedings of the 29th ACM International Conference on Information \& Knowledge Management},
pages = {275–284},
numpages = {10},
location = {Virtual Event, Ireland},
series = {CIKM '20}
}

@inproceedings{rsd_haggai,
author = {Roitman, Haggai and Erera, Shai and Weiner, Bar},
title = {Robust Standard Deviation Estimation for Query Performance Prediction},
year = {2017},
isbn = {9781450344906},
publisher = {Association for Computing Machinery},
address = {New York, NY, USA},
url = {https://doi.org/10.1145/3121050.3121087},
doi = {10.1145/3121050.3121087},
booktitle = {Proceedings of the ACM SIGIR International Conference on Theory of Information Retrieval},
pages = {245–248},
numpages = {4},
location = {Amsterdam, The Netherlands},
series = {ICTIR '17}
}

@inproceedings{uef_kurland_sigir10,
author = {Shtok, Anna and Kurland, Oren and Carmel, David},
title = {Using Statistical Decision Theory and Relevance Models for Query-Performance Prediction},
year = {2010},
booktitle = {Proceedings of the 33rd International ACM SIGIR Conference on Research and Development in Information Retrieval},
publisher = {Association for Computing Machinery},
address = {New York, NY, USA},
pages = {259–266},
series = {SIGIR '10}
}

@article{kurland_tois12,
author = {Shtok, Anna and Kurland, Oren and Carmel, David and Raiber, Fiana and Markovits, Gad},
title = {Predicting Query Performance by Query-Drift Estimation},
year = {2012},
publisher = {Association for Computing Machinery},
volume = {30},
number = {2},
journal = {ACM Trans. Inf. Syst.},
articleno = {11},
numpages = {35},
}

@inproceedings{sen,
  author    = {Procheta Sen and
              Debasis Ganguly},
  title     = {Towards Socially Responsible AI: Cognitive Bias-Aware Multi-Objective Learning},
  booktitle = {Proc. of AAAI 2020},
  pages     = {-},
  year      = {2006},
}

@article{jarvelin2002cumulated,
  title={Cumulated gain-based evaluation of IR techniques},
  author={J{\"a}rvelin, Kalervo and Kek{\"a}l{\"a}inen, Jaana},
  journal={TOIS},
  year={2002}
}

@inproceedings{joachims2005accurately,
  title={Accurately Interpreting Clickthrough Data as Implicit Feedback},
  author={Joachims, Thorsten},
  booktitle={SIGIR},
  year={2005}
}

@inproceedings{Singh2019PolicyLF,
  title={Policy Learning for Fairness in Ranking},
  author={Ashudeep Singh and Thorsten Joachims},
  booktitle={Neural Information Processing Systems},
  year={2019},
  pages={5426-5436},
  url={https://api.semanticscholar.org/CorpusID:60440362}
}

@article{Bower2021IndividuallyFR,
  title={Individually Fair Ranking},
  author={Amanda Bower and Hamid Eftekhari and Mikhail Yurochkin and Yuekai Sun},
  journal={ArXiv},
  year={2021},
  volume={abs/2103.11023},
  url={https://api.semanticscholar.org/CorpusID:232307455}
}

@article{pdd-tois,
author = {Datta, Suchana and Faggioli, Guglielmo and Ferro, Nicola and Ganguly, Debasis and Muntean, Cristina Ioana and Perego, Raffaele and Tonellotto, Nicola},
title = {Projection-Displacement-Based Query Performance Prediction for Embedded Space of Dense Retrievers},
year = {2025},
issue_date = {January 2026},
publisher = {Association for Computing Machinery},
address = {New York, NY, USA},
volume = {44},
number = {1},
issn = {1046-8188},
url = {https://doi.org/10.1145/3765617},
doi = {10.1145/3765617},
journal = {ACM Trans. Inf. Syst.},
month = oct,
articleno = {7},
numpages = {30}
}

@inproceedings{Biega2018EquityOA,
author = {Biega, Asia J. and Gummadi, Krishna P. and Weikum, Gerhard},
title = {Equity of Attention: Amortizing Individual Fairness in Rankings},
year = {2018},
booktitle = {The 41st International ACM SIGIR Conference on Research \& Development in Information Retrieval},
pages = {405–414},
series = {SIGIR '18}
}

@article{Heuss2022FairnessOE,
  title={Fairness of Exposure in Light of Incomplete Exposure Estimation},
  author={Maria Heuss and Fatemeh Sarvi and M. de Rijke},
  journal={Proceedings of the 45th International ACM SIGIR Conference on Research and Development in Information Retrieval},
  year={2022},
  url={https://api.semanticscholar.org/CorpusID:249059216}
}

@inproceedings{Yadav2019PolicyGradientTO,
author = {Yadav, Himank and Du, Zhengxiao and Joachims, Thorsten},
title = {Policy-Gradient Training of Fair and Unbiased Ranking Functions},
year = {2021},
booktitle = {Proceedings of the 44th International ACM SIGIR Conference on Research and Development in Information Retrieval},
pages = {1044–1053},
series = {SIGIR '21}
}

@inproceedings{craswell2008experimental,
author = {Craswell, Nick and Zoeter, Onno and Taylor, Michael and Ramsey, Bill},
title = {An experimental comparison of click position-bias models},
year = {2008},
booktitle = {Proceedings of the 2008 International Conference on Web Search and Data Mining},
pages = {87–94},
series = {WSDM '08}
}

@inproceedings{chapelle2009dynamic,
author = {Chapelle, Olivier and Zhang, Ya},
title = {A dynamic bayesian network click model for web search ranking},
year = {2009},
booktitle = {Proceedings of the 18th International Conference on World Wide Web},
pages = {1–10},
series = {WWW '09}
}

@inproceedings{qpp-bert,
author = {Arabzadeh, Negar and Khodabakhsh, Maryam and Bagheri, Ebrahim},
title = {BERT-QPP: Contextualized Pre-Trained Transformers for Query Performance Prediction},
year = {2021},
publisher = {Association for Computing Machinery},
booktitle = {Proceedings of the 30th ACM International Conference on Information and Knowledge Management},
pages = {2857–2861},
numpages = {5},
location = {Virtual Event, Queensland, Australia},
series = {CIKM '21}
}

@inproceedings{joachims-kdd2018,
author = {Singh, Ashudeep and Joachims, Thorsten},
title = {Fairness of Exposure in Rankings},
year = {2018},
isbn = {9781450355520},
publisher = {Association for Computing Machinery},
address = {New York, NY, USA},
url = {https://doi.org/10.1145/3219819.3220088},
doi = {10.1145/3219819.3220088},
booktitle = {Proceedings of the 24th ACM SIGKDD International Conference on Knowledge Discovery \& Data Mining},
pages = {2219–2228},
numpages = {10},
location = {London, United Kingdom},
series = {KDD '18}
}

@article{mono-duo-lin,
  author       = {Ronak Pradeep and
                  Rodrigo Frassetto Nogueira and
                  Jimmy Lin},
  title        = {The Expando-Mono-Duo Design Pattern for Text Ranking with Pretrained
                  Sequence-to-Sequence Models},
  journal      = {CoRR},
  volume       = {abs/2101.05667},
  year         = {2021},
  url          = {https://arxiv.org/abs/2101.05667},
  eprinttype    = {arXiv},
  eprint       = {2101.05667},
  bibsource    = {dblp computer science bibliography, https://dblp.org}
}

@article{RoyEtAl2019,
  author    = {Dwaipayan Roy and
               Debasis Ganguly and
               Mandar Mitra and
               Gareth J. F. Jones},
  title     = {Estimating Gaussian mixture models in the local neighbourhood of embedded
               word vectors for query performance prediction},
  journal   = {Inf. Process. Manag.},
  volume    = {56},
  number    = {3},
  pages     = {1026--1045},
  year      = {2019},
  url       = {https://doi.org/10.1016/j.ipm.2018.10.009},
  doi       = {10.1016/j.ipm.2018.10.009},
  bibsource = {dblp computer science bibliography, https://dblp.org}
}

@inproceedings{query_variants_kurland,
  author    = {Oleg Zendel and
               Anna Shtok and
               Fiana Raiber and
               Oren Kurland and
               J. Shane Culpepper},
  title     = {Information Needs, Queries, and Query Performance Prediction},
  booktitle = {Proc. of {SIGIR '19}},
  publisher = {Association for Computing Machinery},
  address = {New York, NY, USA},
  pages     = {395--404},
  year      = {2019}
}

@article{DBLP:journals/tois/DattaGMG23,
  author       = {Suchana Datta and
                  Debasis Ganguly and
                  Mandar Mitra and
                  Derek Greene},
  title        = {A Relative Information Gain-based Query Performance Prediction Framework
                  with Generated Query Variants},
  journal      = {{ACM} Trans. Inf. Syst.},
  volume       = {41},
  number       = {2},
  pages        = {38:1--38:31},
  year         = {2023}
}

@inproceedings{DBLP:conf/trec/EkstrandMR022,
  author       = {Michael D. Ekstrand and
                  Graham McDonald and
                  Amifa Raj and
                  Isaac Johnson},
  title        = {Overview of the {TREC} 2022 Fair Ranking Track},
  booktitle    = {{TREC}},
  series       = {{NIST} Special Publication},
  volume       = {500-338},
  publisher    = {National Institute of Standards and Technology {(NIST)}},
  year         = {2022}
}

@inproceedings{DBLP:conf/ecir/SantraBG26,
  author       = {Payel Santra and
                  Partha Basuchowdhuri and
                  Debasis Ganguly},
  title        = {Breaking Flat: {A} Generalised Query Performance Prediction Evaluation
                  Framework},
  booktitle    = {{ECIR} {(2)}},
  series       = {Lecture Notes in Computer Science},
  pages        = {155--171},
  publisher    = {Springer},
  year         = {2026}
}

@inproceedings{santra2025hf,
  title={HF-RAG: Hierarchical Fusion-based RAG with Multiple Sources and Rankers},
  author={Santra, Payel and Ghosh, Madhusudan and Ganguly, Debasis and Basuchowdhuri, Partha and Naskar, Sudip Kumar},
  booktitle={Proceedings of the 34th ACM International Conference on Information and Knowledge Management},
  pages={5202--5207},
  year={2025}
}

@inproceedings{DBLP:conf/ecir/TianGM25,
  author       = {Fangzheng Tian and
                  Debasis Ganguly and
                  Craig Macdonald},
  title        = {Is Relevance Propagated from Retriever to Generator in RAG?},
  booktitle    = {{ECIR} {(1)}},
  series       = {Lecture Notes in Computer Science},
  volume       = {15572},
  pages        = {32--48},
  publisher    = {Springer},
  year         = {2025}
}

@inproceedings{DBLP:conf/ecir/SantraBG26a,
  author       = {Payel Santra and
                  Partha Basuchowdhuri and
                  Debasis Ganguly},
  title        = {Beyond Correlations: {A} Downstream Evaluation Framework for Query
                  Performance Prediction},
  booktitle    = {{ECIR} {(2)}},
  series       = {Lecture Notes in Computer Science},
  volume       = {16484},
  pages        = {518--526},
  publisher    = {Springer},
  year         = {2026}
}

@inproceedings{rpp-gpp,
  author       = {Fangzheng Tian and
                  Debasis Ganguly and
                  Craig Macdonald},
  title        = {Predicting Retrieval Utility and Answer Quality in Retrieval-Augmented
                  Generation},
  booktitle    = {{ECIR} {(1)}},
  series       = {Lecture Notes in Computer Science},
  pages        = {368--385},
  publisher    = {Springer},
  year         = {2026}
}

@inproceedings{DBLP:conf/ecir/DattaGMG24,
  author       = {Suchana Datta and
                  Debasis Ganguly and
                  Sean MacAvaney and
                  Derek Greene},
  title        = {A Deep Learning Approach for Selective Relevance Feedback},
  booktitle    = {{ECIR} {(2)}},
  series       = {Lecture Notes in Computer Science},
  volume       = {14609},
  pages        = {189--204},
  publisher    = {Springer},
  year         = {2024}
}

@inproceedings{wig_croft_SIGIR07,
author = {Zhou, Yun and Croft, W. Bruce},
title = {Query Performance Prediction in Web Search Environments},
year = {2007},
address = {New York, NY, USA},
booktitle = {Proc. 30th Annual International ACM SIGIR Conference on Research and Development in Information Retrieval},
publisher = {Association for Computing Machinery},
pages = {543–550},
series = {SIGIR '07}
}

@inproceedings{FaggioliFerroEtAl2023,
   author    = {Faggioli, Guglielmo and Ferro, Nicola and Muntean, Cristina and Perego, Raffaele and Tonellotto, Nicola},
   title     = {{A Geometric Framework for Query Performance Prediction in Conversational Search}},
   year      = {2023},
   booktitle = {Proceedings of 46th international ACM SIGIR Conference on Research \& Development in Information Retrieval, SIGIR 2023 July 23–27, 2023, Taipei, Taiwan},
   pages        = {1355--1365},
   publisher = {{ACM}},
   doi       = {https://doi.org/10.1145/3539618.3591625}
}

@article{Oh2023ApplyingRL,
  title={Applying Reinforcement Learning for Enhanced Cybersecurity against Adversarial Simulation},
  author={Sang Ho Oh and Min Ki Jeong and Hyung Chan Kim and Jongyoul Park},
  journal={Sensors (Basel, Switzerland)},
  year={2023},
  volume={23},
  url={https://api.semanticscholar.org/CorpusID:257452412}
}

@inproceedings{Feng2022HasCG,
  title={Has CEO Gender Bias Really Been Fixed? Adversarial Attacking and Improving Gender Fairness in Image Search},
  author={Yunhe Feng and C. Shah},
  booktitle={AAAI Conference on Artificial Intelligence},
  year={2022},
  url={https://api.semanticscholar.org/CorpusID:247981991}
}

@inproceedings{DBLP:conf/aaai/SenG20,
  author    = {Procheta Sen and
               Debasis Ganguly},
  title     = {Towards Socially Responsible {AI:} Cognitive Bias-Aware Multi-Objective
               Learning},
  booktitle = {{AAAI}},
  pages     = {2685--2692},
  publisher = {{AAAI} Press},
  year      = {2020}
}

@inproceedings{colbert,
  author       = {Omar Khattab and
                  Matei Zaharia},
  editor       = {Jimmy X. Huang and
                  Yi Chang and
                  Xueqi Cheng and
                  Jaap Kamps and
                  Vanessa Murdock and
                  Ji{-}Rong Wen and
                  Yiqun Liu},
  title        = {ColBERT: Efficient and Effective Passage Search via Contextualized
                  Late Interaction over {BERT}},
  booktitle    = {Proceedings of the 43rd International {ACM} {SIGIR} conference on
                  research and development in Information Retrieval, {SIGIR} 2020, Virtual
                  Event, China, July 25-30, 2020},
  pages        = {39--48},
  publisher    = {{ACM}},
  year         = {2020},
  url          = {https://doi.org/10.1145/3397271.3401075},
  doi          = {10.1145/3397271.3401075},
  bibsource    = {dblp computer science bibliography, https://dblp.org}
}

@inproceedings{jaenich2022university,
  title={University of Glasgow Terrier Team at the TREC 2022 Fair Ranking Track.},
  author={Jaenich, Thomas and McDonald, Graham and Ounis, Iadh},
  booktitle={TREC},
  year={2022}
}

@inproceedings{datta2022pointwise,
  title={A'Pointwise-Query, Listwise-Document'based Query Performance Prediction Approach},
  author={Datta, Suchana and MacAvaney, Sean and Ganguly, Debasis and Greene, Derek},
  booktitle={Proceedings of the 45th International ACM SIGIR Conference on Research and Development in Information Retrieval},
  pages={2148--2153},
  year={2022}
}

@inproceedings{datta2022deep,
  title={Deep-qpp: A pairwise interaction-based deep learning model for supervised query performance prediction},
  author={Datta, Suchana and Ganguly, Debasis and Greene, Derek and Mitra, Mandar},
  booktitle={Proceedings of the fifteenth ACM international conference on web search and data mining},
  pages={201--209},
  year={2022}
}

\end{document}